\documentclass[a4paper,USenglish,cleveref,autoref,numberwithinsect]{lipics-v2021}
\usepackage{xifthen}

\usepackage{mathtools}
\usepackage{commath}
\usepackage{xspace}
\usepackage{bm}

\usepackage{xcolor}

\usepackage{tablefootnote}
\usepackage{comment}
\usepackage{environ}

\theoremstyle{plain}

\let\abs\relax
\DeclarePairedDelimiter\abs{\lvert}{\rvert}
\let\originalleft\left
\let\originalright\right
\def\left#1{\mathopen{}\originalleft#1}
\def\right#1{\originalright#1\mathclose{}}

\newcommand{\floor}[1]{\ensuremath{\lfloor #1 \rfloor}}

\newcommand{\ceil}[1]{\ensuremath{\lceil #1 \rceil}}

\newenvironment{smallbmatrix}{\bigl[\begin{smallmatrix}}{\end{smallmatrix}\bigr]}

\newcommand{\GenFac}{\textsc{GenFac}\xspace}
\newcommand{\CountGenFac}{\#\GenFac}

\newcommand{\Factor}[1]{\ensuremath{#1}-\textsc{Factor}\xspace}
\newcommand{\BFactor}{\Factor{B}}
\newcommand{\MinBFactor}{\textsc{Min}-\BFactor}
\newcommand{\MaxBFactor}{\textsc{Max}-\BFactor}
\newcommand{\CountFactor}[1]{\#\Factor{#1}}
\newcommand{\CountBFactor}{\CountFactor{B}}
\newcommand{\CountMaxBFactor}{\#\textsc{Max}-\BFactor}

\newcommand{\BFactorRelation}{$B$-\textsc{Factor with Relations}\xspace}
\newcommand{\BFR}{$B$-\textsc{Factor}$^{\mathcal R}$\xspace}

\newcommand{\CountBFR}{\#\BFR}

\newcommand{\SETH}{\textsc{Strong Exponential Time Hypothesis}\xspace}
\newcommand{\ETH}{\textsc{Exponential Time Hypothesis}\xspace}

\newcommand{\PerfMatch}{\textsc{PerfMatch}\xspace}

\newcommand{\newextmathcommand}[2]{%
  \newcommand{#1}{{\xspace\ensuremath{#2}\xspace}}
}

\newcommand{\deff}{\coloneqq}

\newextmathcommand{\from}{\leftarrow}

\newcommand{\EQ}[1]{\ensuremath{\mathtt{EQ}_{#1}}}
\newcommand{\HWin}[2][]{\ensuremath{\mathtt{HW}_{\in #2}%
\ifthenelse{\equal{#1}{}}{}{^{(#1)}}}}
\newcommand{\HWeq}[2][]{\ensuremath{\mathtt{HW}_{= #2}%
\ifthenelse{\equal{#1}{}}{}{^{(#1)}}}}
\newcommand{\Sig}{\textup{\texttt{SIG}}}

\newcommand{\hol}[1]{\ensuremath{\operatorname{\mathtt{Holant}}(#1)}}
\newcommand{\Hol}[1]{\ensuremath{\operatorname{Holant}(#1)}}

\DeclareMathOperator{\poly}{poly}
\DeclareMathOperator{\hw}{hw}
\DeclareMathOperator{\maxgap}{max-gap}

\newcommand{\dotcup}{\mathrel{\dot{\cup}}}

\newextmathcommand{\SetB}{\{0,1\}}
\newextmathcommand{\SetF}{\mathbb F}
\newextmathcommand{\SetN}{\mathbb N}
\newextmathcommand{\SetR}{\mathbb R}
\newextmathcommand{\SetZ}{\mathbb Z}

\newcommand{\sharpP}{\textup{\textsf{\#P}}\xspace}
\newcommand{\NP}{\textup{\textsf{NP}}\xspace}

\newcommand{\W}[1]{\textup{\textsf{W}$[#1]$}\xspace}

\renewcommand{\O}{\mathcal O}

\newcommand{\Ostar}[2][n]{\ensuremath{#2 #1^{\mathcal{O}(1)} }}
\newcommand{\tw}{{\operatorname{tw}}}
\newcommand{\pw}{{\operatorname{pw}}}
\newcommand{\cutw}{{\operatorname{cutw}}}

\author{D\'aniel Marx}
{CISPA Helmholtz Center for Information Security, Saarland Informatics Campus, Germany}
{marx@cispa.saarland}
{}
{}
\author{Govind S. Sankar}
{Indian Institute of Technology Madras, Chennai, India}
{govindbose@gmail.com}
{https://orcid.org/0000-0002-7443-9599}
{}
\author{Philipp Schepper}
{CISPA Helmholtz Center for Information Security, Saarland Informatics Campus, Germany \and
Saarbrücken Graduate School of Computer Science, Saarland Informatics Campus, Germany}
{philipp.schepper@cispa.saarland}
{https://orcid.org/0000-0002-5810-7949}
{}
\Copyright{D\'aniel Marx, Govind S.~Sankar, and Philipp Schepper}
\authorrunning{D.~Marx, G.S.~Sankar, P.~Schepper}

\funding{Research supported by the European Research Council (ERC) consolidator grant No.~725978 SYSTEMATICGRAPH.}

\EventEditors{Nikhil Bansal, Emanuela Merelli, and James Worrell}
\EventNoEds{3}
\EventLongTitle{48th International Colloquium on Automata, Languages, and Programming (ICALP 2021)}
\EventShortTitle{ICALP 2021}
\EventAcronym{ICALP}
\EventYear{2021}
\EventDate{July 12--16, 2021}
\EventLocation{Glasgow, Scotland (Virtual Conference)}
\EventLogo{}
\SeriesVolume{198}
\ArticleNo{91}

\graphicspath{{./img/},{../img/}}

\title{Degrees and Gaps:
Tight Complexity Results of General Factor Problems Parameterized by Treewidth and Cutwidth}
\titlerunning{Tight Complexity Results of General Factor Problems}



\nolinenumbers 

\hideLIPIcs

\relatedversion{Full version of the paper accepted for ICALP21.}

\begin{document}
\maketitle

\begin{abstract}
  \renewcommand{\GenFac}{\textsc{General Factor}\xspace}%
 In the \GenFac problem,
  we are given an undirected graph $G$
  and for each vertex $v\in V(G)$ a finite set $B_v$ of non-negative integers.
  The task is to decide if there is a subset $S\subseteq E(G)$
  such that $\deg_S(v)\in B_v$ for all vertices $v$ of $G$.
  Define the $\maxgap$ of a finite integer set $B$ to be the largest $d\ge 0$
  such that there is an $a\ge 0$ with $[a,a+d+1] \cap B = \{a,a+d+1\}$.
  Cornu\'ejols showed in 1988 that if the $\maxgap$ of all sets $B_v$ is at most 1,
then the decision version of \GenFac is polynomial-time solvable.
  This result was extended 2018 by Dudycz and Paluch
  for the optimization (i.e.\ minimization and maximization) versions.
  We present a general algorithm counting the number of solutions of a certain size in time $\Ostar{(M+1)^\tw}$,
  given a tree decomposition of width $\tw$,
  where $M$ is the maximum integer over all $B_v$.
  By using convolution techniques from van Rooij (2020),
  we improve upon the previous $\Ostar{(M+1)^{3\tw}}$ time algorithm by Arulselvan et al.\ from 2018.

  We prove that this algorithm is essentially optimal for all cases
  that are not trivial or polynomial time solvable
  for the decision, minimization or maximization versions.
  Our lower bounds show that such an improvement is not even possible for \BFactor,
  which is \GenFac on graphs where all sets $B_v$ agree with the fixed set $B$.
  We show that for every fixed $B$ where the problem is \NP-hard,
  our $\Ostar{(\max B+1)^{\tw}}$ algorithm cannot be significantly improved:
  assuming the \SETH (SETH),
  no algorithm can solve \BFactor 
  in time $\Ostar{(\max B+1-\epsilon)^\tw}$ for any $\epsilon>0$.
  We extend this bound to the counting version of \BFactor
  for arbitrary, non-trivial sets $B$,
  assuming \#SETH.

  We also investigate the parameterization of the problem by cutwidth.
  Unlike for treewidth, having a larger set $B$ does not appear to make the problem harder:
  we give a $\Ostar{2^\cutw}$ algorithm for any $B$
  and provide a matching lower bound that this is optimal
  for the \NP-hard cases.
\end{abstract}

\section{Introduction}\label{sec:intro}
Matching problems for graphs are widely studied in computer science
\cite{ArulselvanCGMM18,Cornuejols88,CurticapeanM16,DudyczP18,Edmonds65,HoffmannV04,KolisettyLVY19,Lovasz72,MicaliV80,Shiloach81}.
The most prominent ones are \textsc{Perfect~Matching} (\PerfMatch) and \textsc{Maximum-Weight Matching}.
Both problems have long known polynomial-time algorithms \cite{Edmonds65,MicaliV80}
and various generalizations were investigated in the graph-theory literature.
These range from simple extensions such as the $k$-factor problem for a positive integer $k$
(every vertex has to be incident to exactly $k$ edges) \cite{Berge73,KolisettyLVY19},
to more complex ones, where the vertices are assigned intervals \cite{Shiloach81}.
These problems are generally solved by a reduction to \PerfMatch
by replacing the vertices of the original instance with suitable gadgets. 
Lov\'asz introduced a general version of these problems which we call \textsc{General Factor} \cite{Lovasz72}:
\begin{definition}[\textsc{General Factor} (\GenFac)]
  \label{def:generalFactor}
  Let $G=(V,E)$ be an undirected node labeled graph
  where the label of a vertex $v$ is a set $B_v\subseteq \SetN$.
  We say $S \subseteq E$ is a \emph{solution} if
  $\deg_S(v) \in B_v$ for all $v \in V$.
  \GenFac is the problem of deciding whether $G$ has a solution.
\end{definition}
The minimization and maximization versions of \GenFac
are the problems of finding the size of the solution with smallest and largest cardinality, respectively.

\subparagraph*{Polynomial-Time Solvable Cases.}
For several cases (e.g.\ $k$-factor, sets are intervals) reductions to \PerfMatch are known,
leading directly to polynomial-time algorithms.
Cornu\'ejols analyzed the complexity of the general problem
to identify properties of the sets that make the problem easier to solve \cite{Cornuejols88}.
For this he introduced the \emph{gap} of a set:
A gap is a finite sequence of consecutive integers
not contained in the set
but whose boundaries are contained in the set (cf.\ \cref{def:gap}).
For example, the set $\{1,5,6,8\}$ has gaps of size $3$ and $1$.
For a set $S$, $\maxgap S$ denotes the size of its largest gap.
Cornu\'ejols showed that if the max-gaps of all sets are at most $1$,
then the problem is polynomial-time solvable.
Later this result was extended to the maximization and minimization (optimization) versions of \GenFac.
\begin{theorem}[\cite{Cornuejols88,DudyczP18}]
  The decision, maximization, and minimization version of \GenFac can be solved in polynomial time on arbitrary graphs
  if for all nodes $v$, $\maxgap B_v \le 1$.
\end{theorem}
On the other side Cornu\'ejols proved \GenFac to be \NP-complete
if there are nodes with a gap of size two, namely $\{1\}$ and $\{0,3\}$,
by a reduction from \emph{exact $3$-cover}.
More generally, it can be deduced from the work of Feder~\cite{Feder01} that \GenFac becomes \NP-complete whenever every set $B_v$ is restricted to be the same fixed set $B$ having gap size at least two.

\subparagraph*{Treewidth.}
This paper is part of a long sequence of works studying problems parameterized by treewidth
and related metrics like cutwidth or cliquewidth.
Treewidth received significant attention as
many \NP-hard problems like \textsc{Colouring}, \textsc{Independent Set}, or \textsc{Dominating Set} (see \cite{BodlaenderK08} for a survey) are 
polynomial-time solvable on bounded-treewidth graphs.
Courcelle's Theorem \cite{Courcelle90,Courcelle92} shows that a large class of graph problems can 
be solved in linear time on graphs of bounded treewidth. 
Recent developments on the algorithmic side include various techniques such as
Cut \& Count \cite{CyganNPPRW11,Pilipczuk11},
rank-based dynamic programming \cite{BodlaenderCKN13,CyganKN18,FominLS14}
and fast subset convolution \cite{Rooij20,RooijBR09}.
On the negative side, there have been a large number of results showing lower bounds
based on complexity assumptions such as the \ETH (ETH) and the \SETH (SETH) \cite{CaiJ03,CyganDLMNOPSW16,LokshtanovMS11,LokshtanovMS18}.
For many such problems, 
their optimal algorithms utilize some form of dynamic programming, 
where a ``state'' is stored for every node in the tree decomposition.
The number of such states determines the running time of the algorithms, seemingly suggesting that
this number is a natural barrier to the running time of any algorithm.
Typically, the conditional lower bounds confirm this intuition
by showing that no algorithm can break this barrier.

\subparagraph*{New Faster Algorithms.}
One of the first algorithmic results for \GenFac parameterized by treewidth
was given by Arulselvan et al.\ \cite{ArulselvanCGMM18}.
They present an algorithm for a restricted version of the problem
where the sets contain zero and an interval of integers.
This algorithm can be easily extended to handle arbitrary instances
while preserving the running time of $\Ostar{(M+1)^{3\tw}}$
where $M$ is the maximum over all sets assigned to the vertices.
Their algorithm is based on the standard dynamic programming approach when parameterizing by treewidth,
i.e.\ it considers all possible states for each node of the tree decomposition.
The number of states in the dynamic programming is about $(M+1)^{\tw+1}$:
one needs to keep track of the degree of the partial solution at each of the at most $\tw+1$ vertices of a bag of the tree decomposition,
and this degree can be between 0 and $M$.
Therefore, a natural question is whether the algorithm can be improved
to obtain an $\Ostar{(M+1)^\tw}$ running time, matching the number of states.
Such improvements are known for other problems,
for example for \textsc{Dominating Set} and \#\PerfMatch in \cite{RooijBR09}.
We base our algorithm on the same dynamic programming idea, 
but instead of processing all combination of states at join nodes,
we make use of the technique of van Rooij~\cite{Rooij20} to compute fast convolutions,
avoiding this bottle-neck of the computation.
The algorithm can be easily generalized to the optimization and counting versions as well; to unify the results, we present the algorithm in a way that counts all solutions of a certain given size.
\begin{theorem}\label{thm:algo:treewidth:main}
  Given a \GenFac instance $G$ and a tree decomposition of width $\tw$.
  Let $M = \max_{v\in V(G)} \max B_v$.
  Then for all $s$, we can count the solutions of size exactly $s$
  in time $\Ostar{(M+1)^{\tw}}$.
\end{theorem}
As we shall see, this algorithm is essentially optimal for every fixed $B$ where \BFactor is \NP-hard.
Note that in order to obtain this optimal running time,
we have to use a well-known, but non-trivial technique;
beyond that, our algorithm does not provide new insights into the problem.

\subparagraph*{Tight Lower Bounds for the Decision Version.}
To investigate how the properties of the sets $B_v$ influence the complexity of the problem,
we give conditional lower bounds based on SETH for the restrictive \BFactor problem,
where all sets have to be the same fixed set $B$.
By a careful design our 
lower bounds also hold for a parameterization by \emph{pathwidth}.
Note that if the set is not fixed,
Arulselvan et al.\ showed that \GenFac is \W1-hard when parameterizing only by 
treewidth \cite{ArulselvanCGMM18}.
Thus, it is reasonable to focus only on the cases
with fixed sets to prove tight lower bounds.
\begin{theorem}[Lower Bound for Decision Version]
  \label{thm:lower:dec}
  Let $B\subseteq \SetN$ be a fixed, finite set with $0 \notin B$ and $\maxgap B > 1$.
  If, given a path decomposition of width $\pw$, 
  \BFactor can be solved in time $\Ostar{(\max B+1-\epsilon)^{\pw}}$ for some $\epsilon>0$
  even on graphs with degree at most $2\max B$,
  then SETH is false.
\end{theorem}
The same result immediately follows for treewidth
as for all graphs the pathwidth forms an upper-bound for the treewidth \cite{ParameterizedAlgos}.
Hence, our algorithm is optimal not only for \GenFac but also for \BFactor parameterized by treewidth and does not allow major improvements.

\subparagraph*{Tight Lower Bounds for the Optimization Version.}
It suffices to consider the maximization version with $\maxgap B>1$ and $0\in B$ for the optimization version.
The other cases are either polynomial-time solvable ($\maxgap B\le 1$ or $0\in B$ for \MinBFactor) \cite{DudyczP18}
or the hardness directly follows from the lower bound for the decision version.
Observe that the assumption $0\in B$ does not make the problem trivially solvable.
For these cases, we give essentially the same lower bound as for the decision version.
Again the bound rules out that we can improve the given algorithm substantially;
the running time is essentially optimal.
\begin{theorem}[Lower Bound for Maximization Version]
  \label{thm:lower:maximization}
  Let $B\subseteq \SetN$ be a fixed, finite set
  with $\maxgap B > 1$.
  If, given a path decomposition of width $\pw$, 
  \MaxBFactor can be solved in time $\Ostar{(\max B+1-\epsilon)^{\pw}}$ for some $\epsilon>0$
  even on graphs with degree at most $2\max B$,
  then SETH is false.
\end{theorem}

\subparagraph{Counting.}
It is well known that \PerfMatch can be solved in polynomial time \cite{Edmonds65}.
Surprisingly, Valiant showed in \cite{Valiant79_permanent}
that \emph{counting} the number of perfect matchings of a graph
is as hard as counting satisfying assignment of a boolean formula.
This is curious
as (presumably) no polynomial-time algorithm for the decision version of the latter problem exists.
The observation then led to the definition of the complexity class \sharpP
containing the counting problems
whose corresponding decision version lies in \NP.
Indeed, this feature that some structures are easy to find but hard to count appears in our work as well.
Apart from \#\PerfMatch, which itself is \CountFactor{\{1\}},
our results imply that \CountBFactor is \sharpP-hard for any finite, fixed $B$.
This contrasts with the decision version, where the problem is easy when $\maxgap B\leq 1$.
Over and above showing \sharpP-hardness, we show a tight lower bound for \CountBFactor, assuming \#SETH, the counting version of SETH.
There have been several results \cite{Curticapean16,CurticapeanM16,DellHMTW14} based on \#SETH and \#ETH.
Some of our constructions were inspired by one such work by Curticapean and Marx \cite{CurticapeanM16},
where they show a lower bound of $(2-\epsilon)^{\pw}n^{\O(1)}$ for \#\PerfMatch on graphs assuming \#SETH.
We prove a wide generalization of this result by providing a tight lower bound for every \CountBFactor problem.
As for the optimization and decision version,
our algorithm shows the tightness of this lower bound.

\begin{theorem}[Lower Bound for Counting Version]
  \label{thm:lower:counting}
  Let $B\subseteq\SetN$ be a nonempty, fixed, and finite set such that $B\neq \{0\}$.
  If, given a path decomposition of width $\pw$, 
  \CountBFactor can be solved in time $\Ostar{(\max B+1-\epsilon)^{\pw}}$ for some $\epsilon>0$
  even on graphs with degree at most $2\max B+6$,
  then \#SETH is false.
\end{theorem}
We also investigate \CountMaxBFactor, the problem of counting maximum-sized solutions. 
The following argument shows that \CountFactor{\{\max B\}} can be reduced to \CountMaxBFactor without increasing pathwidth,
hence \cref{thm:lower:counting} gives a lower bound of $\Ostar{(\max B+1-\epsilon)^{\pw}}$.
Consider an instance of \CountFactor{\{\max B\}} on a graph $G$ of pathwidth $\pw$. 
In polynomial time, check if $G$ has some \Factor{\{\max B\}} \cite{Cornuejols88}. 
If it does not, then output 0. 
If it does, then solve \CountMaxBFactor on $G$. As now every maximum-sized \Factor{B} is actually a \Factor{\{\max B\}}, this indeed solves the
\CountFactor{\{\max B\}} problem.
\begin{corollary}
	Let $B\subseteq \SetN$ be a fixed, finite set such that $B\neq \{0\}$. 
	If, given a path decomposition of width $\pw$, 
	\CountMaxBFactor can be solved in time $\Ostar{(\max B+1-\epsilon)^{\pw}}$ for some $\epsilon>0$
	even on graphs with degree at most $2\max B+6$,
	then SETH is false.
\end{corollary}
We leave open the question of a tight lower bound for the minimization version.

\subparagraph*{Parameterizing by Cutwidth.}
As previously mentioned, pathwidth and treewidth are not the only parameters used in parameterized complexity.
Cutwidth, cliquewidth, genus, and crossing number are only a few more examples of a vast class of possible parameters.
For cutwidth,
we consider linear layouts of graphs with $n$ vertices,
which are just enumerations $v_1,\dots,v_n$ of all graph vertices.
Then the cut after vertex $v_i$ consists of all edges in $G$ with one end in $\{v_1,\dots,v_i\}$ and the other end in $\{v_{i+1},\dots,v_n\}$.
The \emph{cutwidth} of a linear layout is the maximum over the size of the cut after every $v_i$.
The cutwidth $\cutw$ of a graph is the minimum over the cutwidths of all possible linear layouts.
As $\tw\le\pw\le\cutw$,
it is not completely surprising that we get different upper bounds for cutwidth.
But now
a simple dynamic program suffices to prove the upper bound for this case.
Further, the running time of the algorithm is independent from the maximum of the set $B$.
\begin{theorem}\label{thm:algo:cutwidth:main}
  Given a linear layout of a \GenFac instance $G$ with width $\cutw$, 
  for all $s$
  we can count 
  the number of solutions of size exactly $s$ 
  in time $\Ostar{2^\cutw}$,.
\end{theorem}
This again matches the number of states for each cut of the linear layout. 
Like before, we omit the algorithm here and refer the reader to the full version. 
Note that the running times appearing in Theorems~\ref{thm:algo:treewidth:main} and \ref{thm:algo:cutwidth:main} cannot be directly compared: the base is lower when parameterized by cutwidth, but cutwidth can be larger than treewidth.

By a modified high-level construction, we show
matching lower bounds based on SETH for the decision and optimization versions,
and, based on \#SETH, for the counting version.
\begin{theorem}[Lower Bounds for Cutwidth]\label{thm:lowerCW:main}
  Let $B\subseteq \SetN$ be a fixed, nonempty set of finite size.
  If, given a linear layout of width $\cutw$,
  the following problems can be solved in time $\Ostar{(2-\epsilon)^{\cutw}}$ for any $\epsilon>0$ 
  even on graphs with degree at most $2\max B +6$,
  then SETH (resp.\ \#SETH) fails:
  (1)~\BFactor and \MinBFactor if $0 \notin B$ and $\maxgap B > 1$,
  (2)~\MaxBFactor if $\maxgap B > 1$,
  and (3)~\CountBFactor if $B \neq \{0\}$.
\end{theorem}

\subsection{Techniques}
\subparagraph{\NP-Hardness.}
Existing results that explicitly show hardness for the \GenFac (or \BFactor) problem have been limited.
For example, Cornu\'ejols showed in \cite{Cornuejols88} the hardness only for the case when the lists are either $\{0,3\}$ or $\{1\}$.
However, as we discuss below, the results from Feder in \cite{Feder01} implicitly show that the \BFactor problem is \NP-hard
when the allowed lists have a gap of size more than 1.
Feder~\cite{Feder01} considered the complexity of Boolean CSP instances where each variable appears in two constraint and each constraint enforces a relation $R$.
As pointed out by Dalmau and Ford in \cite{DalmauF03},
the results in \cite{Feder01} show that when $R$ does not form a $\Delta$-matroid, then the problem is \NP-hard.

We say that the relation $R$ is {\em symmetric} if membership of a tuple in $R$ depends only on the Hamming weight of the tuple, i.e., $R$ contains exactly those tuples that have Hamming weight from a set $B$. In this case, a CSP instance where each variable appears in two $R$-constraints can be naturally represented as a \BFactor instance: the vertices correspond to the constraints and the edges correspond to the variables.\footnote{There is a minor technicality here: the resulting graph can be a multigraph. This can be resolved by placing vertices with list $\{0,2\}$ on the parallel edges. For the case where $\maxgap B>1$, such vertices can be \emph{realized} (see \cref{lem:intro:realization}) by vertices with list $B$. }
It can be observed that if $B$ has gaps of size more than 1, then the corresponding symmetric relation $R$ cannot form a $\Delta$-matroid.
This, combined with the results from \cite{Feder01} 
show the \NP-completeness of \BFactor when $\maxgap B>1$.

\subparagraph{High-level Structure.}
The known \NP-hardness proofs do not fully utilize the capacity for \GenFac to represent relations.
For example, \cite{Feder01} argues that just being able to represent the relation $\{(x,y,z): x=y=z\}$ is enough to ensure \NP-hardness.
Our lower bounds use the fact that \GenFac (and even \BFactor) can represent almost \emph{any} relation.
We have two main steps.
We first reduce from SAT to the intermediate problem \BFactorRelation(\BFR).
This is an extension of \BFactor where we can additionally assign (almost) arbitrary relations to the vertices of the graph
and not only the set $B$ (cf.\ \cref{def:bfr}).
This reduction step is similar to previous lower bound constructions \cite{LokshtanovMS18}. 
In the second step, we reduce from \BFR to \BFactor by replacing each relation and its node by a graph expressing the required relation.

For both the pathwidth and cutwidth lower bounds, our broad-strokes idea is the same. 
See \cref{fig:intro:high-level} for an illustration.
For a given SAT instance, we use vertical columns of high-degree vertices to store the assignment (or ``state'') of the variables.
This state is stored as the number of edges to its left that are selected in any solution. 
When parameterizing by cutwidth, this process is easier;
we can afford one vertex for each variable of the SAT instance
since we are only looking for a lower bound of $\Ostar{(2-\epsilon)^{\cutw}}$.
However, in the pathwidth case,
we want to show a lowerbound of $\Ostar{(\max B+1-\epsilon)^\pw}$.
This means that our constructed graph should have pathwidth $\frac{n}{\log(\max B+1)} +\O(1)$.
We achieve this by grouping $q$ variables together and associate a set of $g$ vertices with each group,
where $2^q$ is roughly $(\max B+1)^g$. 
By this the lower bound we get is
the number of assignments or states that we can store for a given pathwidth. 
This directly connects to the
dynamic-programming algorithms on path decompositions, 
where the running time is determined by how many states each node has to store. 
Thus, our vertical columns represent several of the bags of the path decomposition.

\begin{figure}
	\centering
	\includegraphics[scale=0.7]{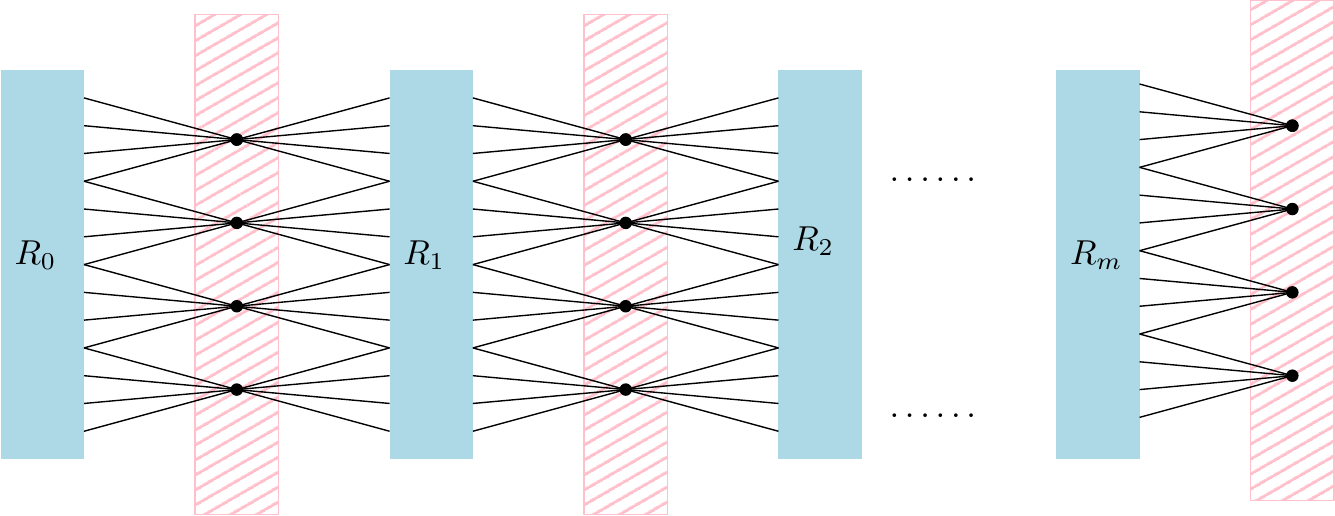}
	\caption{Pink boxes represent the vertical columns of vertices that store the state of the variables.
  Blue boxes labelled $R_i$ verify that the $i$th clause is satisfied.
  Here, $m$ is the number of clauses.}
	\label{fig:intro:high-level}
\end{figure}

Once we are able to propagate the information about the assignments through these vertical columns,
we use some relations placed between them to verify that the assignment satisfies all the clauses. 
Afterwards, we exploit the (large) gap of the set $B$
and replace the relations by vertices with list $B$
such that we finally get a \BFactor instance.

\subparagraph{Realizations for the Decision Version.}
We do this step from \BFR to \BFactor
through what we call \emph{realizations}. 
We say that we can realize a relation $R \subseteq \SetB^d$
if we have a gadget $G$ with $d$ \emph{dangling edges}, 
exactly one corresponding to each coordinate of the relation,
such that $G$ behaves in some special way.
These dangling edges (say $D$) have one endpoint in $G$ and another endpoint outside of $G$.
These edges act as inputs to the function $f$.
We say that the function is realized
iff the only solutions are those that pick $D'\subseteq D$ such that $D' \in R$.
For a more rigorous definition, see \cref{sec:dec}.
For example, when showing the hardness for the decision version of the problem, we use the following key result.
We say that a Boolean function is \emph{even} if it is zero for all inputs with odd Hamming weight.

\begin{lemma}\label{lem:intro:realization}
  Let $R \subseteq 2^{[d]}$ be an even Boolean relation.
  Then for every fixed, finite set $B\subseteq \SetN$
  with $\maxgap B>1$ and $0\notin B$,
  there is a \BFactor instance $G$ with a set $D$ of $d$ dangling edges
  identified with $[d]$
  such that
  for all $D' \subseteq D$ it holds that
  $D' \in R$
  if and only if
  there is a solution $S \subseteq E$ with $S \cap D = D'$.
\end{lemma}
See our later sections for a more formal treatment of this statement.
This method of realizing relations was originally introduced in \cite{CurticapeanM16}
to show tight lower bounds for \#\PerfMatch.
We also show similar lemmas for the optimization and counting versions.

\subparagraph*{Realizations for the Optimization Version.}
As mentioned earlier, for the maximization version
we have to deal with $0\in B$
which implies that we cannot rule out solutions in all cases.
Therefore, we relax the condition ``no solution'' for the realization
and allow the existence of ``small'' solutions.
This gap between largest (as we are considering the maximization version)
and small solutions will be the \emph{penalty} of the realizations.
To understand how we implement this penalty,
first observe that the crucial difference between the decision and the optimization version is that we cannot force edges.
Hence, it almost suffices to find a way to force edges to a vertex and obtain a large solution.
Otherwise, if the edges we want to force are not selected,
the solution should be small.
For this we exploit the large gap of the set $B$
and combine it with regular graphs of high girth.%
\footnote{The girth of a graph is the length of the shortest cycle of a graph.}
The idea is now the following:
If the forced edge is selected, we can select all edges.
Otherwise, another (internal) edge is not selected,
because we cannot have a degree not contained in $B$.
As the girth of the graph is large, this will lead to many not selected edges and hence the penalty of this realization is large.

\subparagraph{Holants.}
For the counting version,
we first describe the Holant problem for (possibly weighted) graphs.
Introduced by Valiant in \cite{Valiant04},
Holants captures a large class of graph problems including coloring problems, vertex cover, and matchings \cite{KowalczykC16}.
The input to the problem is a graph and relations over the incident edges of each vertex.
The Holant problem on unweighted graphs asks to count the number of subgraphs such that each such relation is satisfied.
This can be seen as a general version of \CountBFR
where all vertices are assigned a relation.
Hence, \CountGenFac can be seen as the restriction of the Holant problem on unweighted graphs 
where all relations are \emph{symmetric}.
See \cref{sec:count} for a formal introduction.

There has been a significant amount of work on the Holant framework in the past two decades \cite{CaiHL12,CaiLX11,GuoL17,Lu11}.
A few of these works directly imply a \sharpP-hardness for some cases of \CountBFactor.
For example, we can use the dichotomy theorem from Cai, Huang, and Lu \cite{CaiHL12} in the following way when $B$ contains 0, 1 and 2:
Let $x$ be the smallest number not in $B$.
To a vertex of list $B$, attach $x-3$ vertices with list $B$ and attach three dangling edges.
This gadget can be seen as one node with three dangling edges.
It can be shown via the dichotomy theorem by Cai et al.\ %
that the Holant problem is \sharpP-hard when only such nodes are allowed.
However, it is not immediate how to realize relations with this approach, nor is it clear how the treewidth of the graphs change through the holographic reductions that are common in showing \sharpP-hardness with Holant problems.
We circumvent these problems by showing the realizations through polynomial interpolation instead of holographic reductions.
To our knowledge,
this technique was introduced by Valiant in \cite{Valiant79},
where the problem of counting \emph{all} matchings was shown to be as hard as counting \emph{perfect} matchings.
In our notation, this would correspond to an interpolation argument between \CountFactor{\{1\}} and \CountFactor{\{0,1\}}.

\subparagraph*{Techniques for Cutwidth Lower Bound.}
As mentioned earlier, we additionally show a matching lower bound parameterized by cutwidth.
For this it is crucial to see that the construction for the pathwidth lower bound has large cutwidth.
Hence, instead of having several variables per vertex,
each vertex now corresponds to exactly one variable.
Although this construction does not work for the pathwidth bound as it produced large pathwidth,
it is sufficient in this setting.
The graph will be a grid-like structure
where the edges of the rows encode the assignment to the variables
and the edges of the columns simulate whether we have already satisfied the clause.
We start with every clause being unsatisfied in the first row.
Then for each row we check at the crossing points
if the variable and the corresponding assignment satisfy the clause.
We define relations of constant degree to simulate this behaviour.
Then we make use of the previous realization results and replace all relations by their realization.
This almost immediately gives us the lower bounds when parameterized by cutwidth.

\subsection{Structure of the Paper}
\cref{sec:prelim} introduces useful notation and definitions we use in the paper.
The two algorithms for parameterization by treewidth and cutwidth are given in \cref{sec:algo}.
In \cref{sec:lower} we show the reduction from SAT to our intermediate problem \BFR.
The version-dependent constructions of the realizations and the corresponding lower bounds are shown in \cref{sec:dec} for the decision version,
in \cref{sec:opt} for the optimization version,
and for the counting version in \cref{sec:count}.
\cref{sec:lowerCW} is dedicated to the lower bound construction when parameterizing by cutwidth.

\section{Preliminaries}\label{sec:prelim}
We introduce homogeneous graphs
to formally define \BFactor.
\begin{definition}[Homogeneous Graphs and \BFactor]
  \label{def:homogeneousGraphs}
  \label{def:B-Factor}
  Let $B \subseteq \SetN$ be some fixed, finite set.
  We say a node-labeled graph is \emph{$B$-homogeneous}
  if for each node $v \in V$ it holds that $B_v = B$.
  Then \BFactor is the restriction of \GenFac to $B$-homogeneous graphs.
\end{definition}
This definition
directly transfers to the optimization and counting version.
We now formally introduce the max-gap of integer sets
along with some other properties.
\begin{definition}\label{def:gap}
  Let $B \subseteq \SetN$ be finite.
  We define $\maxgap B$ as the largest non-negative integer $d$
  such that there is an $a \in B$ with $[a, a+d+1] \cap B = \{a,a+d+1\}$.
\end{definition}
In this paper we regularly insert graphs into other graphs.
To make this operation formal, we make use of \emph{dangling edges}: these are edges that have only one endpoint.
We denote a dangling edge with endpoint $v$ by $(?,v)$.
For the sake of completeness we now formally define this procedure of replacing the relations,
i.e.\ the insertion of a graph into another graph.
\begin{definition}[Insertion]\label{def:insertion}
  Let $G=(V, E)$ be a graph
  and $v \in V$ be of degree $k$,
  with incident edges $e_1=(v_1,v), \dots, e_k=(v_k,v)$
  that are ordered in some fixed way.
  Let $H=(W, F)$ be a graph
  with dangling edges $d_1=(?,u_1), \dots, d_k=(?,u_k)$
  where the $u_i$ are not necessarily pairwise distinct.
  \emph{Inserting} $H$ in $G$ at $v$ gives us a new graph $G'=(V',E')$ where:
  \[V' = (V \cup W) \setminus \{v\}
    \text{~~and~~}
    E' = (E \cup F) \setminus \{e_1, \dots, e_k, d_1, \dots, d_k\} \cup
      \{ (v_1, u_1), \dots, (v_k, u_k)\}
  \]
\end{definition}
All lower bounds we prove in this paper are based on the \SETH.
But instead of using the original statement
we use a formulation which is more useful to work with.
\begin{conjecture}[\SETH (SETH) \cite{CalabroIP09, ImpagliazzoP01}]
  \label{conj:seth}
  For all $\delta >0$,
  there is a $k \ge 3$ such that satisfiability of $k$-CNF formulas on $n$ variables
  requires more than $(2-\delta)^n$ time.
\end{conjecture}

\subparagraph*{Treewidth and Pathwidth.}
We use the standard definition of a tree decomposition to define treewidth,
as for example given in Chapter~7 of \cite{ParameterizedAlgos}.
A tree decomposition of a undirected graph $G$ is a pair $\mathcal T = (T, \{X_t\}_{t\in V(T)})$,
where $T$ is a tree,
and every node $t$ of $T$
\footnote{We refer to the vertices of $T$ as \emph{nodes}, to distinguish them from the vertices of $G$.}
is assigned a \emph{bag} $X_t \subseteq V(G)$,
such that
\begin{itemize}
  \item
  Every vertex of $G$ is contained in at least one bag:
  $\bigcup_{t\in V(T)} X_t = V(G)$.
  \item
  For every $(u,v) \in E(G)$,
  there exists a node $t$ of $T$ such that $X_t$ contains both $u$ and $v$.
  \item
  For every $u \in V(G)$,
  the nodes whose corresponding bags contain $u$,
  i.e.\ $T_u = \{t \in V(T) \mid u \in X_t \}$,
  induces a connected subtree of $T$.
\end{itemize}
The \emph{width} of the tree decomposition is
$\max_{t\in V(T)}\abs{X_t}-1$,
i.e.\ the maximum size of the bags minus 1.
The minimum possible width of a tree decomposition of a graph $G$
defines the \emph{treewidth} $\tw(G)$ of $G$.
The path decomposition leading to pathwidth $\pw(G)$ uses the same definition,
except that we require that $T$ is a path.

We use nice tree decompositions as they help in simplifying algorithms;
more precisely,
we use a variant of nice decompositions with introduce edge nodes.
They are defined by a tree decomposition $\mathcal T= (T, \{X_t\}_{t\in V(T)})$ 
and additionally one vertex $r$ of $T$ which serves as the root of $T$.
This rooted tree decomposition $\mathcal T$ is \emph{nice} if
the bag of the root is empty, i.e.\ $X_r = \emptyset$,
and all nodes $t$ of $T$ have one of the following types:
\begin{itemize}
  \item
  \textbf{Leaf node:}
  $t$ has no children and the bag is empty,
  i.e.\ $X_t=\emptyset$.
  \item
  \textbf{Introduce vertex node:}
  $t$ has a unique child $t'$
  with $X_t = X_{t'} \cup \{v\}$ for some vertex $v \notin X_{t'}$.
  \item
  \textbf{Introduce edge node:}
  $t$ is labeled with an edge $(u,v) \in E(G)$ where $u, v \in X_t$.
  Further $t$ has a unique child $t'$ such that $X_t = X_{t'}$.
  We say, edge $(u,v)$ is introduced at $t$.
  \item
  \textbf{Forget node:}
  $t$ has a unique child $t'$
  with $X_t = X_{t'}\setminus \{w\}$ for some vertex $w \in X_{t'}$.
  \item
  \textbf{Join node:}
  $t$ has exactly two children $t_1$ and $t_2$
  such that $X_t = X_{t_1} = X_{t_2}$.
\end{itemize}
Given a tree decomposition, it can efficiently be transformed into a nice tree decomposition
by increasing its size at most polynomially, \cite{Kloks94}.
Hence, we always assume w.l.o.g.\ that a given tree decomposition is nice.

\subparagraph*{About Relations.}
A relation $R: \SetB^k \to \SetB$
can also be seen as a set $R' \subseteq \SetB^k$
such that $x\in R'$ iff $R(x) = 1$.
We can also identify $R$ with a set $R'' \subseteq 2^{[k]}$,
where each element of $R''$ contains the positions of the 1s of an accepted input.
Precisely, $x'=\{i\mid x[i]=1\} \in R''$ iff $R(x) = 1$.
We switch between these definitions depending on the context.
Recall, a relation is \emph{symmetric} if its output only depends on the Hamming weight of its input.

To simplify notation,
we introduce the following generic classes of symmetric relations.
\begin{definition}\label{def:relations}
  For a vector $x\in \SetB^k$, we define $\hw(x)$ as the number of 1s in $x$,
  i.e.\ the Hamming weight of $x$.
  We define the following
  for $S \subseteq \SetN$, and $i,j\in \SetN$:
  \[
    \HWin[j]{S} \deff \{ x \in \SetB^j \mid \hw(x) \in S\}
    \qquad
    \EQ{j} \deff \HWin[j]{\{0,j\}}
    \qquad
    \HWeq[j]{i} \deff \HWin[j]{\{i\}}
  \]
  $\EQ{j}$ is the \emph{equality} relation on $j$ inputs.
  We use $\HWin{S}$ to denote $\HWin[j]{S}$ when the arity $j$ of the relation is implicit.
  We also use this as the set of the relations $\HWin[j]{S}$ for all $j \in \SetN$.
  We transfer this abuse of notation to $\HWeq{i}$.
\end{definition}
Note that assigning the relations $\HWin{B}$ to a vertex
corresponds to assigning the set $B$ to the vertex.
Which notation is used depends on the context we are in.

\section{Algorithms}\label{sec:algo}

In this section we present two algorithms for \GenFac parameterized by treewidth and cutwidth, respectively.
By our lower bounds for \BFactor we know that these algorithms are (conditionally) optimal for both problems.
This is somewhat surprising as the intuitive feeling would be
that we can exploit that the graphs are $B$-homogeneous for \BFactor.

Let $G$ be the given graph in the following
and let $B_v\subseteq \SetN$ be the label of vertices $v\in V(G)$.
We can assume $B_v \subseteq \{0,\dots,\abs{V(G)}\}$.
Then, let $M = \max_{v \in V} \max B_v$ in the following.

\subsection{Algorithm when Parameterizing by Treewidth}
\label{sec:algo:treewidth}
Assuming a tree decomposition of width $\tw$ is given,
Arulselvan et al.\ showed implicitly that \GenFac can be solved in time $\Ostar{(M+1)^{3\tw}}$ \cite{ArulselvanCGMM18}.
We base our optimal $\Ostar{(M+1)^{\tw}}$ algorithm for counting solutions of a certain size on this approach
and combine it with a more precise runtime analysis and fast convolution techniques.
As for other algorithms parameterized by treewidth
we also assume that a tree decomposition is given as input.
This is reasonable as it allows studying the problem independently from finding a decomposition.

\subsubsection{Algorithm}
Let $k$ be the width of the given tree decomposition
where every edge is introduced exactly once.
For a node $t$ in the tree decomposition
let $X_t$ be the corresponding bag of vertices.
We define $V_t$ to be the set of vertices
and $E_t$ to be the set of edges introduced in the subtree rooted at $t$.

For a fixed node $t$
let $f:X_t \to [0,n]$
and $s \in [0,m]$ where $m$ is the number of edges of $G$.
We define $c[t, f, s]$ to be the number of sets $F \subseteq E_t$ with $\abs{F} = s$
such that $\deg_{F}(v) = f(v)$ for all $v \in X_t$
and $\deg_{F}(v) \in B_v$ for all $v \in V_t\setminus X_t$.

As the bag of the root $r$ is empty,
only the empty assignment is valid
and hence the value of $c[r,\emptyset, s]$ corresponds to the number of solutions of size exactly $s$ for all $s\in [0,m]$.
For all other $s$ the number of solutions is $0$.

The following dynamic program computes all entries of $c$
and its behavior is completely determined by the type of the node $t$ it currently visits.
\begin{description}
\item[Leaf node]
We have $X_t = \emptyset$
and set $c[t, \emptyset, 0] \deff 1$ and 0 for all other entries.

\item[Introduce Vertex Node]
Let $t'$ be the unique child of $t$ such that $X_t = X_{t'} \cup \{v\}$.
Let $f\vert_{X_{t'}}$ be the restriction of $f$ to domain $X_{t'}$.
That is $f\vert_{X_{t'}} : X_{t'} \to [0,n]$ and $f\vert_{X_{t'}}(x) = f(x)$ for all $x \in X_{t'}$.
We set
\[
c[t, f, s] \deff
\begin{cases}
  c[t', f\vert_{X_{t'}}, s] & \text{if } f(v) = 0 \\
  0 & \text{else}
\end{cases}
\]

\item[Introduce Edge Node]
Let $(u,v)$ be the edge introduced at node $t$.
We have to decide whether to pick the edge or not.
As the edge can only be chosen if the degree of the vertices $u$ and $v$ is not zero,
we define:
\[
c[t, f, s] \deff c[t', f, s] + \begin{cases}
  c[t', f_{u \mapsto f(u)-1, v \mapsto f(v)-1}, s-1] & \text{if } f(u)f(v) > 0 \\
  0 & \text{else}
\end{cases}
\]

\item[Forget Node]
Let $t'$ be the child of $t$ such that $X_t = X_{t'}\setminus \{v\}$.
Then we set:
\[ c[t, f, s] \deff \sum_{d \in B_v} c[t', f_{v \mapsto d}, s] \]

\item[Join Node]
Let $t_1$ and $t_2$ be the left and right child of $t$ such that $X_t = X_{t_1} = X_{t_2}$.
For three functions $f, f_1, f_2 :X_t \to [0,n]$ we write $f_1+f_2 = f$ if we mean $f_1(v)+f_2(v) = f(v)$ for all $v \in X_t$.
Based on this, we can define the values of the table as follows:
\[
  c[t, f, s] \deff
  \sum_{\substack{f_1, f_2 \text{ s.t.} \\ f_1+f_2=f}}
  \sum_{s_1 + s_2 = s}
  c[t_1, f_1, s_1] \cdot c[t_2, f_2, s_2]
\]
Observe that for this to work we crucially need that each edge is introduced exactly once in the whole decomposition.
Otherwise the edge could be picked in both subtrees leading to a wrong solution since it could be counted twice.
But by assumption each edge is either contained in the tree rooted at $t_1$ or $t_2$ or neither of them and thus above $t$.
\end{description}

\subparagraph*{Running Time.}
One can simply observe that each entry of the table can be computed in time $\max\{ \O(1), M, (n+1)^{2k+2}m^2\}$.
But we can actually restrict the codomain of $f$ to $[0, M]$
and it suffices for the join node to pick $f_1$ and $s_1$ as $f_2$ and $s_2$ then follow from the equations $f_1+f_2=f$ and $s_1+s_2=s$, respectively.
Hence, each cell can be computed in time $\O((M+1)^{k+1} m)$.
Since for each node there are at most $(M+1)^{k+1}$ many valid functions $f$,
the problem can be solved in time $\Ostar{(M+1)^{2k}}$.
In the next section we show how to improve the running time of the join node by computing all values in parallel.

\subsubsection{Improving the Running Time}
For the improved computation of the join node
we make use of the concepts and ideas from van Rooij \cite{Rooij20} including the following lemma:
\begin{lemma}[Lemma~3 in \cite{Rooij20} (cyclic and non-cyclic convolution)]\label{lem:cyclicNonCyclicConvolution}
  Let $N = \SetN_{<q_1} \times \SetN_{<q_2} \times \dots \times \SetN_{<q_l}$, and let $Q = \prod_{i=1}^l q_i$.
  \footnote{We set $\SetN_{<r} = \{0,\dots,r-1\}$.}
  Let $Z = \SetZ_{r_1} \times \SetZ_{r_1} \times \dots \times \SetZ_{r_k}$,
  and let $R = \prod_{i=1}^k r_i$.

  Let $f,g : Z \times N \to \SetF_p$, where $p$ is chosen appropriately.
  And, let $h : Z\times N \to \SetF_p$ be the combined (partially cyclic and partially non-cyclic) convolution of $f$ and $g$ defined as:
  \[
  h(x, i) = \sum_{y_1+y_2 \equiv x} \sum_{j_1+j_2=i} f(y_1,j_1) g(y_2,j_2)
  \]
  where the sum $y_1+y_2 \equiv x$ is evaluated component-wise modulo $r_i$ at coordinate $i$ (sum in $Z$),
  and the sum $j_1+j_2=i$ is evaluated component-wise without modulus (sum in $N$).
  Then, the combined convolution $h$ can be computed in $\O(R Q 2^l(\log(R) + \log(Q) + l))$ arithmetic operations.
\end{lemma}
The obvious way of using the lemma
would be to directly use the non-cyclic part of the convolution.
But this would result in $Q=(M+1)^{k+1}(m+1)$ and $l=k+2$,
which gives us a bound of $\Ostar{(M+1)^k2^k}$.
Instead we mainly use the cyclic part of the convolution and add a degree bound in the non-cyclic part
such that we only consider correct combinations of values.
\begin{lemma}
  For each join-node $t$ of the tree decomposition,
  we can compute $c[t, f, s]$ for all valid $f$ and all $s\in[0,m]$ with $\O((M+1)^k k^2Mm\log(Mm))$ arithmetic operations.
\end{lemma}
\begin{proof}
  We follow the ideas of the proof of Lemma~10 in \cite{Rooij20}.
  Let $t_1$ and $t_2$ be the two children of node $t$.
  Since $\abs{X_t} \le k+1$,
  we can transform each function $f:X_t \to [0,M]$ into a vector in $[0,M]^{k+1}$ (and vice versa if we add dummy vertices to $X_t$).

  Based on this we define $a, b: [0,M]^{k+1} \times \SetN \times [0,m] \to \SetF_p$ where $p > 2^{\abs{E}}$
  since there can be at most that many edge subsets of the graph $G$.
  We define:
  \[
  a(f, F, s) \deff
    \begin{cases}
      c[t_1, f, s] & \text{if }\norm{f}_1 = F \\
      0 & \text{else}
    \end{cases}
  \quad
  \text{ and }
  \quad
  b(f, F, s) \deff
    \begin{cases}
      c[t_2, f, s] & \text{if }\norm{f}_1 = F \\
      0 & \text{else}
    \end{cases}
  \]
  Where we set $\norm{f}_1 = \sum_{v \in X_t} f(v)$,
  where the sum is over the integers (without modulo).
  Hence, it actually suffices to restrict $\SetN$ in the definition of $a$ and $b$ to $\SetN_{\le (k+1)M}$.
  We use this to compute the following convolution for all $f\in[0,M]^{k+1}$, $F\in[0,(k+1)M]$, and $s\in[0,m]$.
  \[
    c(f, F, s) =
    \sum_{f_1+f_2 \equiv f}
    \sum_{\substack{F_1+F_2=F \\s_1+s_2=s}}
    a(f_1,F_1,s_1) b(f_2,F_2,s_2)
  \]
  Then we set $c[t, f, s] = c(f, \norm{f}_1, s)$ for all $f$ and $s$ to complete the computation for this join-node.

  Again the computation of $f_1+f_2 \equiv f$ is component wise modulo $M+1$ and for $F_1+F_2=F$ and $s_1+s_2=s$ over the integers without modulo.
  By the definition of $a$ and $b$ we only sum values if $\norm{f_i}=F_i$.
  Thus we actually just sum over all values $f_1, f_2$ such that $f_1+f_2\bm=f$
  because otherwise we would get that $\norm{f_1}_1+\norm{f_2}_1 > \norm f_1$.
  We can directly apply \cref{lem:cyclicNonCyclicConvolution}
  with $l=2$, $q_1=(k+1)M+1$, $q_2=m+1$, and $r_1 = \dots = r_{k+1} = M+1$.
  As $Q=(kM+M+1)(m+1)$ and $R = (M+1)^{k+1}$ the claimed running time follows.
\end{proof}
It is easy to see that the computation of the join nodes dominates the running time.
Therefore, the proof of this lemma finishes the proof of \cref{thm:algo:treewidth:main}.

\subsection{Algorithm when Parameterizing by Cutwidth}\label{sec:algo:cutwidth}

Assume we are given a linear layout of the \GenFac instance $G$.
We relabel the vertices such that their index corresponds to this layout.
Let $C_i$ be the edges of the graph that cross the cut after vertex $v_i$,
that is $C_i = E \cap (\{v_1,\dots,v_i\}\times\{v_{i+1},\dots,v_n\})$.
We define $\delta^+_i= C_i \cap I(v_i)$ as the forwards edges of $v_i$,
i.e.\ the edges to vertices in $\{v_{i+1}, \dots,v_n\}$
and likewise $\delta^-_i= C_{i-1} \cap I(v_i)$ as the backwards edges of $v_i$.
Let $E_i = E\vert_{v_1,\dots,v_i}$ be the edges induced by the vertices $v_1,\dots,v_i$.
We additionally define $\widehat C_i$ as the set of edges in the cuts $C_{i-1}$ and $C_{i}$ that are not incident to $v_i$,
i.e.\ $\widehat C_i = C_{i-1}\setminus \delta^-_i = C_i \setminus \delta^+_i$.
See \cref{fig:algo:cut:notation} for an illustration of the setting.
\begin{figure}
  \centering
  \includegraphics{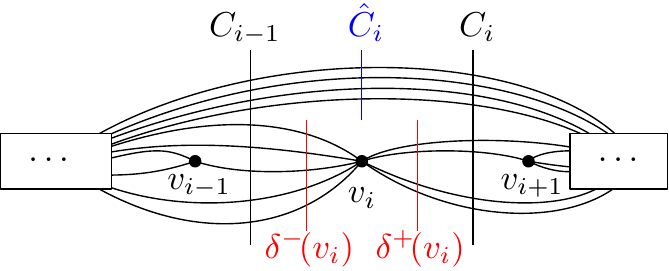}
  \caption{Illustration how the different sets are related to each other.}
  \label{fig:algo:cut:notation}
\end{figure}

The idea of the algorithm is to fill a three-dimensional table $c$ such that
for all $i\in[0,n]$, $s\in[0,m]$, and $C \subseteq 2^{C_i}$
we have that $c[i, C, s]$ is the number of sets $S \subseteq E_i$ with $\abs{S}=s$
such that $\deg_{S}(v_j) + \deg_C(v_j) \in B_{v_j}$ for all $j \in [i]$.
Phrased differently: the number of possibilities to extend the selected edges of the cut
to a valid partial solution for the first $i$ vertices.

As the cut after $v_n$ does not contain any edges,
we only have to consider the empty set.
Hence, the output of the algorithm is the value of $c[n,\emptyset,s]$
for $s\in[0,m]$ and 0 otherwise.

It remains to show how we can compute the values of the table efficiently.
We initialize the table with $c[0,\emptyset,0] = 1$ and 0 at all other positions.
Then we design a step function that computes the values of $c[i,\cdot,\cdot]$ given the values of $c[i-1,\cdot,\cdot]$.
We first show a naive way for this
and then improve the computation to get the algorithm with optimal running time.

\subsubsection{Naive Step Function}
For all $C \subseteq C_{i}$ and $s\in[0,m]$ we compute the value sequentially as follows:
\[
  c[i, C, s] =
  \sum_{\substack{S^- \subseteq \delta^-_i \text{ s.t.}\\ \abs{S^-} + \deg_{C}(v_i) \in B_{v_i}}}
  c[i-1, C \setminus \delta^+_i \cup S^-, s-\abs{S^-}]
\]
For each $i$ we have to compute $\O(2^\cutw(m+1))$ table entries.
Each of the entries can be computed in time $\Ostar{2^\cutw}$ as we iterate over all subsets of $\delta^-_i$.
This gives us a total running time of $\Ostar{4^\cutw}$.

\subsubsection{Improved Step Function}
For the improved algorithm we make use of the fact
that for the computation of $c[i,\cdot,\cdot]$ only the selected edges in $\delta^-_i$
and the \emph{number} of selected edges from $\delta^+_i$ are of interest.
The edges in $\widehat C_i$ are not affected as $v_i$ moves to the other side of the cut.
Further observe that we can partition $C\subseteq C_i$ into $C=\widehat C \dotcup S^+$
for $S^+ \subseteq \delta^+_i$ and $\widehat C \subseteq \widehat C_i$.
We transform the naive formula from above as follows:
\begin{align*}
  c[i, C, s]
  &= c[i, \widehat C \dotcup S^+, s] \\
  &=
  \sum_{\substack{S^- \subseteq \delta^-_i \text{ s.t.}\\ \abs{S^-} + \deg_{\widehat C \dotcup S^+}(v_i) \in B_{v_i}}}
  c[i-1, (\widehat C \dotcup S^+) \setminus \delta^+_i \cup S^-, s-\abs{S^-}] \\
  &=
  \sum_{\substack{S^- \subseteq \delta^-_i \text{ s.t.}\\ \abs{S^-} + \abs{S^+} \in B_{v_i}}}
  c[i-1, \widehat C \cup S^-, s-\abs{S^-}] \\
  &= h[\abs{S^+},\widehat C,s] \\
  \intertext{where:}
  h[\ell,\widehat C,s] &=
  \sum_{\substack{S^- \subseteq \delta^-_i \text{ s.t.}\\ \abs{S^-} + \ell \in B_{v_i}}}
  c[i-1, \widehat C \cup S^-, s-\abs{S^-}]
\end{align*}
The improved step function now works as follows:
\begin{enumerate}
  \item 
  Initialize the auxiliary table $h$ with 0s.
  \item
  Use the above formula to fill $h$ for all
  $\widehat C \subseteq \widehat C_i$,
  $s \in [0,m]$,
  and $\ell \in [0,M]$.
  This takes time $\Ostar{2^{\abs{\widehat C_i}} (M+1) 2^{\abs{\delta^-_i}}} \le \Ostar{2^{\abs{C_{i-1}}}}$.
  \item
  For all
  $\widehat C \subseteq \widehat C_i$,
  $s \in [0,m]$,
  and $S^+ \subseteq \delta^+_i$,
  we set
  \(
    c[i, \widehat C\cup S^+, s] \deff h[\abs{S^+}, \widehat C, s]
  \).
  This takes time $\Ostar{2^{\abs{\widehat C_i}} 2^{\abs{\delta^+_i}}} \le \Ostar{2^{\abs{C_i}}}$.
\end{enumerate}
Hence, we can bound the running time of each round by $\Ostar{2^\cutw M}$.

\section{Lower Bound when Parameterizing by Pathwidth}\label{sec:lower}
We show the lower bounds in two steps.
The first step is a reduction from CNF-SAT
to the intermediate problem \BFactorRelation. 
In the second step, we reduce to the actual version of \BFactor for which we want to show the lower bound.
As the lower bounds are for pathwidth
they immediately hold for treewidth as it can be upper-bounded by pathwidth for all graphs.
\begin{definition}[\BFactorRelation (\BFR)]\label{def:bfr}
  Let $B\subseteq \SetN$ be fixed of finite size.
  $G=(V_S\dotcup V_C, E)$ is an instance of \emph{\BFactorRelation} 
  if all nodes in $V_S$ are labeled with set $B$
  and all nodes $v \in V_C$ are labeled with a relation $R_v$
  that is given as a truth table
  such that the following holds:
  \begin{enumerate}
    \item Let $I(v)$ be the set of edges incident to $v$ in $G$.
    Then $R_v \subseteq 2^{I(v)}$.
    \item There is an even $c_v>0$
    such that for all $x \in R_v$ we have $\hw(x)=c_v$.
  \end{enumerate}
  A set $\widehat E \subseteq E$ is a \emph{solution} for $G$ if
  (1) for $v\in V_S$: $\deg_{\widehat E}(v) \in B$
  and
  (2) for $v \in V_C$: $I(v)\cap \widehat E \in R_v$.
  \BFR is the problem of deciding if such an instance has a solution.
  \\
  We call $V_S$ the set of \emph{simple} nodes and $V_C$ the set of \emph{complex} nodes.
\end{definition}
Using this intermediate problem, we can formally state the first part of the reduction.
The lower bound needs a careful formulation:
when we reduce \BFR to \BFactor by inserting gadgets realizing the relations at the complex nodes,
the size and pathwidth of the graph can increase significantly.
Therefore, we state a stronger lower bound that can tolerate additional terms to take care of such increases.
The key point is that this increase is mainly influenced by the total degree of the complex nodes in a bag of the path decomposition.

\begin{theorem}\label{thm:lower:satToBFR}
  Let $B \subseteq \SetN$ be a fixed set of finite size with $B\neq\{0\}$.
  Given a \BFR instance along with a path decomposition of width $\pw$
  such that $\Delta^* = \max_{\text{bag } X} \sum_{v \in X\cap V_C} \deg(v)$.
  Assume \BFR can be solved in $\Ostar{(\max B+1-\epsilon)^{\pw+f_B(\Delta^*)}}$ time
  on graphs with $n$ vertices  for some $\epsilon > 0$ 
  and some function $f_B:\SetN\to\SetR^+$ that may depend on the set $B$.
  Then SETH fails.
\end{theorem}

\subparagraph*{High Level Idea.}
We follow the ideas of previous lower bound reductions from \cite{LokshtanovMS18}
and combine them with the concept of using relations from \cite{CurticapeanM16}.
From now on let $M\deff\max B$.
Let $\phi$ be the given CNF formula with $n$ variables and clauses $C_1,\dots,C_m$.
Instead of encoding each variable separately,
we group $q$ variables together and encode (partial) assignments to these groups.
For each partial assignment, we define a vector in $[0,M]^g$,
where $g$ is chosen such that $2^q \le (M+1)^g$.
For each group we define a \emph{layer} with $g$ parallel \emph{rows},
where each row corresponds to one dimension of the vector.
The layers consist of an alternation of $g$ parallel simple nodes and a complex node that is related to a clause.
All simple nodes are connected to their neighboring complex nodes by $M$ parallel edges.
The vector from above then corresponds to the number of selected edges
from a simple node to the following shared complex node.
The complex nodes check whether the assignment represented by the selected edges of a layer
satisfies the related clause.
For each clause we connect the related complex nodes by a path.
This path is used to propagate the information whether the clause is already satisfied by some partial assignment
or whether it still needs to be satisfied.
We ensure that each clause is initially not satisfied and eventually all clauses must be satisfied.

\begin{figure}[t]
  \centering
  \includegraphics{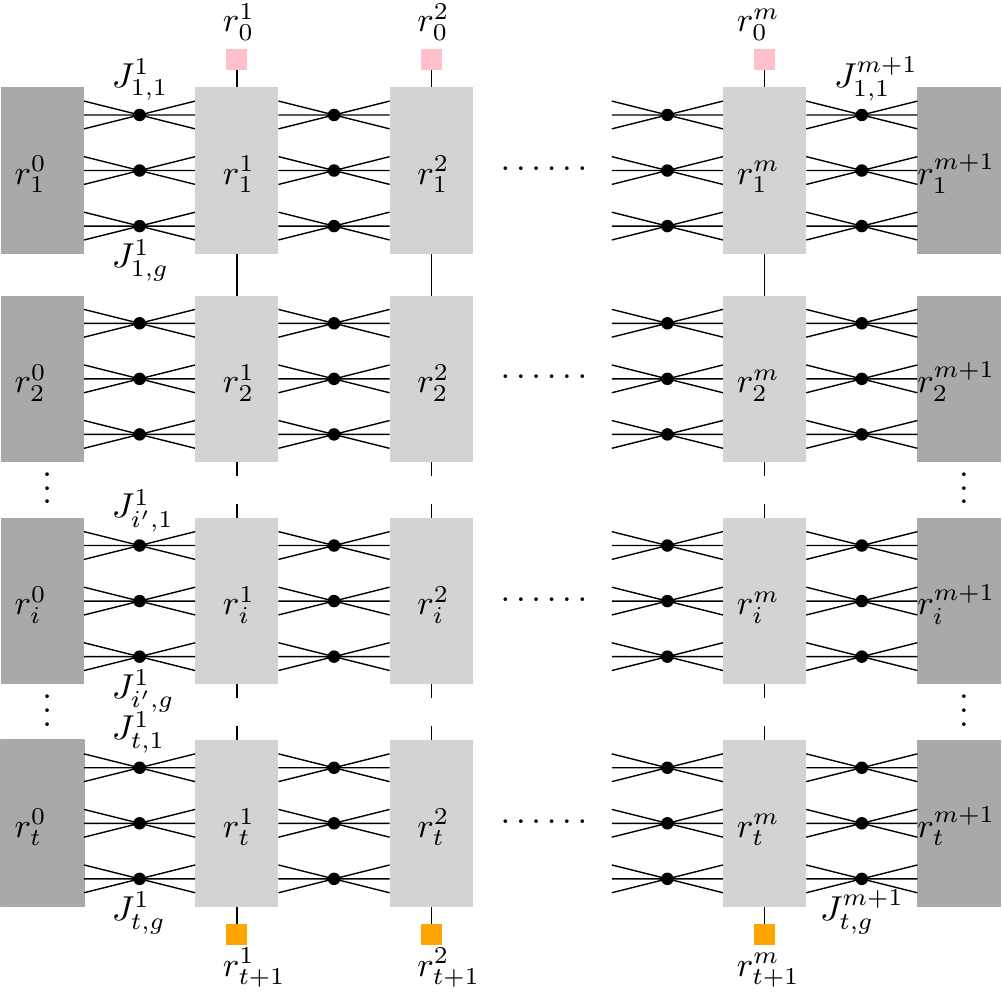}
  \caption{An example illustrating the construction from the proof of \cref{thm:lower:satToBFR}.
  The simple nodes are represented by circles while the complex nodes are represented by boxes.}
  \label{fig:lower:construction}
\end{figure}

\subparagraph*{Constructing the \BFR Instance.}
See \cref{fig:lower:construction} for an example of the following construction.
Split the variables of $\phi$ into $t \deff \ceil{n/q}$ groups $F_1,\dots,F_t$ of size at most $q$, where $q$ is chosen later.
For each of the $t$ groups we encode the $2^q$ partial assignments by vectors from $[0,M]^g$ for some $g$ chosen later.
Instead of using all $(M+1)^g$ possible encodings
we only use those vectors where the total weight of the coordinates is equal to $gM/2$
(we will choose $g$ as a multiple of 4, hence $gM/2$ is an even number).
It can easily be shown
that there are more than $(M+1)^g/(gM+1)$ vectors with exactly this weight. 
Thus, after setting $q = \floor{\log((M+1)^g)-\log(gM+1)}$,
we can map each of the $2^q$ assignments of a group $F_i$ to a distinct vector $[0,M]^g$ with weight exactly $gM/2$.
We say that an partial assignment $\tau$ to a group $F_i$ satisfies a clause $C_j$
if at least one literal in the clause is satisfied under the assignment $\tau$.
Note, that a group $F_i$ does not have to cover all variables of $C_j$ to satisfy the clause.

We define the graph now as follows:
\begin{enumerate}
  \item
  For all $i\in[t]$, $\ell\in[g]$, and $j\in[m]$:
  create a simple node $J_{i,\ell}^j$.

  \item
  For all $i\in[t]$ and $j\in[m]$:
  create complex nodes $r_i^j$
  with relation $R_i^j$ to be defined later.

  \item
  For all $i\in[t]$:
  create complex nodes $r_i^0$ and $r_i^{m+1}$
  with relation $R^0$.

  \item
  For all $j\in[m]$:
  create complex nodes $r_0^j$ (resp.\ $r_{t+1}^j$)
  with relation $\HWeq{0}$ (resp.\ $\HWeq{1}$).

  \item
  For all $i\in[t]$, $\ell\in[g]$, and $j\in[m]$:
  make $J_{i,\ell}^j$ adjacent to $r_{i-1}^j$ and $r_i^j$
  by $M$ parallel edges each.
  We call these edges backwards and forwards edges, respectively.

  \item
  For all $i\in[t]$ and $j\in[m]$,
  make $r_i^j$ additionally adjacent to $r_{i-1}^j$ and $r_{i+1}^j$ by one edge each.
  The degree of the nodes is now $2gM+2$.
\end{enumerate}
We call the set of nodes $\{r_i^j, J_{i,\ell}^j\}_{j,\ell}$ the $i$th \emph{layer}.
The set $\{J_{i,\ell}^j\}_j$ forms the $\ell$th \emph{row} of the $i$th layer.
For a fixed $j\in[m]$,
the set of nodes $\{r_i^j\}_i$ is called the $j$th \emph{column}.

The idea is now the following:
For each partial assignment $\tau$ to a group $F_i$,
we define a vector $v_\tau \in [0,M]^g$ of weight $gM/2$ as its \emph{encoding}.%
\footnote{Note that for different groups the encoding of the same partial assignment do not need to be the same.}
Then $v_\tau[\ell]$ corresponds to the number of selected forward edges of the simple nodes in the $\ell$th row of the $i$th layer.
The vertical edges encode whether a clause was already satisfied.
That is, if the edge between $r_i^j$ and $r_{i+1}^j$ is selected,
then there is some group $F_k$ with $k \le i$ where the corresponding assignment satisfies the clause $C_j$.
By the relation of the nodes $r_0^j$ every clause is initially not satisfied.
But the relation of the $r_{t+1}^j$ nodes ensures that every clause is eventually satisfied.

\subparagraph*{Defining the Relations.}
$R^0$ accepts exactly those inputs of Hamming weight exactly $gM/2$,
an even number by assumption,
where the selected edges for each row must precede the unselected edges,
i.e.\ the first $k$ edges are selected, the next $M-k$ are not selected.

The relation $R_i^j\subseteq \SetB^{2Mg+2}$ of node $r_i^j$ is defined as follows:
\begin{itemize}
  \item
  For $\ell\in[g]$,
  let $x_\ell$ (resp.\ $y_\ell$)
  be the number of selected incident edges to $J_{i,\ell}^j$ (resp.\ $J_{i,\ell}^{j+1}$).
  \item
  $\sum_{\ell \in [g]} x_\ell = gM/2 = \sum_{\ell \in [g]} y_\ell$.
  \item
  $\langle x_1,\dots,x_g \rangle$ describes a valid encoding,
  i.e.\ it corresponds to a partial assignment for $F_i$.
  \item
  $x_\ell + y_\ell = M$.
  Further, 
  the $x_\ell$ (resp.\ $y_\ell$) selected edges
  precede the $M-x_\ell$ (resp.\ $M-y_\ell$) unselected edges
  of the $M$ parallel edges going to a simple node.
  \item
  If the ingoing top edge is selected,
  then the outgoing bottom edge is also selected.
  \item
  If the ingoing top edge is not selected:
  \begin{itemize}
    \item 
    If $C_j$ does not contain a variable of $F_i$,
    then the outgoing bottom edge is not selected.
    \item
    If $C_j$ contains at least one variable of $F_i$,
    then the outgoing bottom edge is selected
    if and only if the selected edges correspond to a valid partial assignment satisfying $C_j$.
  \end{itemize}
\end{itemize}

\subparagraph*{Final Modifications.}
Unfortunately the previous construction is not a \BFR instance.
One reason is that the Hamming weight of the accepted inputs of the relations $R_i^j$
is not always the same even number.
The misbalance is caused by the edges encoding the evaluation of the clause.
To avoid this problem,
we add another edge between all $r_i^j$ and $r_{i+1}^j$ nodes.
We refer to this new edge as the negated edge
and to the original edge as the positive edge.
We further introduce nodes $\hat r_0^j$ and $\hat r_{t+1}^j$
with relations $\HWeq{1}$ and $\HWeq{0}$
which are connected to $r_1^j$ and $r_t^j$, respectively.
Observe that the relations are exactly opposite to the ones for the nodes $r_0^j$ and $r_{t+1}^j$.
We extend all relations $R_i^j$ by two new inputs and 
require that the negated edge is selected if and only if the corresponding positive edge is not selected.
Then, the Hamming weight of the accepted inputs of the relations $R_i^j$ is always the same,
namely $gM+2$, which is an even number as $gM/2$ is even by assumption. 
In a second step we merge the four nodes $r_0^j$, $r_{t+1}^j$, $\hat r_0^j$, and $\hat r_{t+1}^j$
into a single node $r^j$ of degree 4, for all $j\in [m]$.
The relation of $r^j$ accepts only the unique input of weight 2
which agrees with the relations of all four replaced nodes.
We later show that the parallel edges disappear
when replacing the complex nodes by their realizations
(the portal nodes of the realization will be pairwise different).
Nevertheless, if one wants to get a simple graph,
one can use five appropriately connected $\HWeq[4]{2}$ nodes
to replace one such edge.
We leave it to the reader to find the gadget.
Almost all of these modification are only necessary due to parity issues.
As they are in most cases not relevant for the correctness of the construction,
we mostly ignore them in our proofs in the following.

\begin{lemma}\label{lem:lower:construction:completeness}
  If $\phi$ is satisfiable, then there is a solution to the \BFR instance.
\end{lemma}
\begin{proof}
  Let $\sigma$ be a satisfying assignment to the variables of $\phi$.
  For fixed variable group $F_i$,
  let $v_i\in [0,M]^g$ be the encoding that corresponds to $\sigma$ when restricted to $F_i$.
  For the solution we select the edges as follows:
  In the $\ell$th row of the $i$th layer we select the $v_i[\ell]$ top most forwards edges of the simple nodes
  and the $M-v_i[\ell]$ top most backwards edges.
  By this obviously every simple node is incident to $M$ selected edges
  and the conditions of the relations $R^0$ are satisfied.

  Now consider a fixed clause $C_j$.
  Let $F_k$ be the first (with regard to $k$) group of variables 
  whose partial assignment satisfies the clause.
  Then we do not select the edges between $r_{i-1}^j$ and $r_{i}^j$ for all $i\in[k]$.
  Combined with the valid encoding from above,
  the relations of these nodes are now satisfied.
  As the encoding precisely corresponds to the assignment satisfying the clause,
  we select the edges between the nodes $r_{i}^j$ and $r_{i+1}^j$ for all $i\ge k$.
  Then all relations are satisfied and the claim follows.
\end{proof}

\begin{lemma}\label{lem:lower:construction:correctness}
  If there is a solution to the \BFR instance, then $\phi$ is satisfiable.
\end{lemma}
\begin{proof}
  Let $S\subseteq E$ be a solution.
  We construct an assignment $\sigma$ by assigning values to the variables of each group $F_i$ independently.
  For this fix some $i\in[t]$.

  For $\ell\in[g]$,
  let $x_\ell^1,\dots,x_\ell^{m+1}$ and $y_\ell^1,\dots,y_\ell^{m+1}$
  be the number of selected forwards and backwards edges
  of the nodes $J_{i,\ell}^j$, respectively.
  We claim that the number of selected forwards edges does not change for a fixed row $\ell \in[g]$:
  \begin{claim}\label{claim:lower:construction:correctness:helper}
    $x_\ell^1=\dots=x_\ell^{m+1}$
    and $y_\ell^1=\dots=y_\ell^{m+1}$.
  \end{claim}
  \begin{claimproof}
    By the construction of the graph and the definition of the relations we get:
    \begin{align}
      y_\ell^{j+1} + x_\ell^j &= M  && \forall \ell\in [g], j\in[m] \\
      y_\ell^j+x_\ell^j &\le M && \forall \ell\in [g], j\in[m+1] \\
      \sum_{\ell=1}^g x_\ell^{m+1} &= \frac{gM}2 = \sum_{\ell=1}^g y_\ell^1 \\
      \intertext{Combining (1) and (2) gives us:}
      x_\ell^j &\ge x_\ell^{j+1} && \forall j\in[m] \\
      y_\ell^j &\le y_\ell^{j+1} && \forall j\in[m]
    \end{align}
    Now assume for contradictions sake that there is some $j'$ with $x_\ell^{j'} > x_\ell^{j'+1}$:
    \[
    \frac{gM}2
     \overset{(3)}  = \sum_{\ell=1}^g x_\ell^{m+1}
     \overset{(4)}\le \sum_{\ell=1}^g x_\ell^{j'+1}
                    < \sum_{\ell=1}^g x_\ell^{j'}
     \overset{(4)}\le \sum_{\ell=1}^g x_\ell^1
     \overset{(2)}\le \sum_{\ell=1}^g (M-y_\ell^1)
                    = gM - \sum_{\ell=1}^g y_\ell^1
     \overset{(3)}  = \frac{gM}2
    \]
    which is obviously a contradiction.
    This proves the claim as we can use the analogous argument for the backwards edges. 
    Observe that this implies that all simple nodes have degree exactly 
    $M$ in the solution. 
  \end{claimproof}

	Since $S$ is a solution, the relation $R_{i}^{j}$ must be satisfied.
	Thus, $\langle x_1^1,\dots,x_g^1 \rangle$ must correspond to a valid encoding of some assignment for $F_i$.
  We use this partial assignment for $\sigma$.

  To show that $\sigma$ satisfies all clauses
  let $C_j$ be an arbitrary clause.
  As $S$ is a solution, the relations of the nodes $r_0^j$ and $r_{t+1}^j$ must be satisfied by $S$.
  By definition of the relations $R_i^j$,
  there must be a $i_j$ and a node $r_{i_j}^j$
  where the top edge is not selected
  but the bottom edge is selected.
  By definition of the relation, this change can only happen
  if the selected incident forwards edges describe an encoding
  which corresponds to an assignment satisfying $C_j$.
  By \cref{claim:lower:construction:correctness:helper} this number of selected edges is always the same.
  Hence, the assignment represented by the selected edges that are incident to $r_i^j$
  agrees with the assignment we chose for $\sigma$.
\end{proof}
To obtain a tight lower bound,
we need to analyze the pathwidth of our construction and have to bound the degree of the complex nodes.
\begin{lemma}\label{lem:lower:construction:bounds}
  The graph has $\O(tgm)$ simple and $\O(tm)$ complex nodes.
  The degree of the complex nodes is bounded by $2gM+4$.
  The degree of the simple nodes is bounded by $2M$.
  We can efficiently construct a path decomposition of width $tg + \O(1)$
  where at most three complex nodes are simultaneously in one bag.
\end{lemma}
\begin{proof}
  The number of nodes and their degree bound follows immediately from the construction.
  We give a mixed search strategy as shown in \cite{ParameterizedAlgos} to create a path decomposition of the graph.

  We start by placing searchers on each node $J_{i,\ell}^1$ for $i\in [t]$ and $\ell\in[g]$.
  We place a searcher on each $r_i^0$ and remove it immediately afterwards.
  Then we clean the rest of the graph in stages, starting with stage 1.
  Stage $j$ starts if the searchers are placed on the nodes $J_{i,\ell}^j$.
  First place two searcher on $r^j$ and $r_1^j$.
  Then we move the searchers placed on the adjacent simple nodes of $r_1^j$
  to the next simple node of the row.
  Then place one searcher on $r_2^j$ and repeat the process after removing the one from $r_1^j$.
  Repeat this procedure for all other nodes $r_i^j$.
  When removing the searcher from $r_t^j$, we also remove the searcher from $r^j$.
  By this we eventually arrive at a point where all searchers are located at the nodes $J_{i,\ell}^{m+1}$.
  Then we use one additional searchers to clean the remaining $r_i^{m+1}$ nodes.
\end{proof}
Now we have everything ready to prove the lower bound for the intermediate problem \BFR based on the previous construction.
Recall, we defined $\Delta^*$
as the maximum total degree of the complex nodes appearing in one bag,
that is $\Delta^* = \max_{\text{bag } X} \sum_{v \in X\cap V_C} \deg(v)$.
\begin{proof}[Proof of \cref{thm:lower:satToBFR}]
  Assume an $\Ostar[N]{(M+1-\epsilon)^{\pw+f(\Delta^*)}}$ time algorithm for \BFR
  on graphs with $N$ nodes exists for some function $f$.
  Due to rounding issues and the fact that we only use even relations,
  the parameters must be chosen quite carefully.
  We set $\lambda \deff \log_{M+1}(M+1-\epsilon) < 1$.
  Choose an $\alpha > 1$ such that $\alpha \cdot \lambda = \delta' = \log(2-\delta)<1$ for some $\delta>0$.
  This is always possible since $\lambda < 1$.
  Choose $g$ large enough such that $g \log(M+1) \le \alpha \floor{g \log(M+1) - \log(gM+1) }$ and $g$ is divisible by 4.
  By \cref{lem:lower:construction:bounds} and our choice of parameters
  ($t = \ceil{n/q}$ and $q = \floor{g \log (M+1) - \log(gM+1)}$)
  we get:
  \begin{align*}
    &\Ostar[N]{(M+1-\epsilon)^{\pw+f(\Delta^*)}}
     = \Ostar[N]{(M+1-\epsilon)^{tg + \O(1)+ f(\Delta^*)} } \\
    &= \Ostar[N]{(M+1-\epsilon)^{\ceil{\frac n q} g + f(\Delta^*)} }
     = \Ostar[N]{(M+1-\epsilon)^{\frac {n \cdot g} {\floor{g\log(M+1) - \log(gM+1)} } + g + \O(1) + f(\Delta^*)} }
  \end{align*}
  Since $B$ is fixed, $M$ is constant. 
  Further, $g$ only depends on $M$ and $\epsilon$ 
  and $N$ only depends on the number of variables $n$ 
  and number of clauses $m$ of the SAT instance.
  As $\Delta^* \in \O(gM)$,
  the factor of $(M+1-\epsilon)^{\O(1)+g+f(\Delta^*)}$
  contributes only a large constant (depending only on $M$ and $\epsilon$) to the overall running time
  which can be hidden by the $n^{\O(1)}$ term:
  \begin{align*}
  &\le \Ostar[(n+m)]{(M+1-\epsilon)^{ \frac {\alpha n} {\log(M+1)} }} \\
    &= \Ostar[(n+m)]{ 2^{\log(M+1-\epsilon) \frac {\alpha n}{\log(M+1)} }} 
     = \Ostar[(n+m)]{ 2^{\lambda \alpha n }}
     = \Ostar[(n+m)]{ 2^{\delta' n } } \\
     &= \Ostar[(n+m)]{ (2-\delta)^n }.
  \end{align*}
  This violates SETH (\cref{conj:seth})
  since this gives an algorithm to solve $k$-SAT with $n$ variables and
  $m$ clauses in time $\Ostar[(n+m)]{(2-\delta)^n}$ for some constant $\delta>0$
  for all $k$.
\end{proof}

For the lower bound for the counting version,
we need a one-to-one correspondence between valid assignments and solutions to the graph problem.
This is why we need that the selected edges are at specific positions,
i.e.\ above the not-selected edges.
\begin{corollary}\label{corr:lower:construction:parsimonious}
  The reduction in Theorem~\ref{thm:lower:satToBFR} is parsimonious, i.e.\ it preserves the number of solutions.
  Therefore, the statement also holds for the counting versions of \BFR and SETH.
\end{corollary}
We also make the following observations about the constructed \BFR instances.
\begin{corollary}
  \label{corr:lower:construction:observations}
The lower bound proved in \cref{thm:lower:satToBFR} holds even for \BFR
instances where 
\begin{enumerate}
	\item the vertices in $V_S$ form an independent set,
	\item every vertex in $V_S$ is adjacent to exactly two vertices
  from $V_C$ by $\max B$ edges each,
	\item and every vertex in $V_S$ will have degree exactly
	$\max B$ in the solution. 
\end{enumerate}
\end{corollary}

\section{Decision Version}\label{sec:dec}
In this section we prove the lower bound for the decision version of \BFactor
by a reduction from the intermediate \BFR problem.
For this we formally define the concept of realizations
and show that we can realize all relations of a \BFR instance.
Replacing the nodes and their relations by these realizations
yields the final lower bound.
\begin{definition}[Realization]\label{def:dec:realization}
  Let $R \subseteq \SetB^k$ be a relation.
  Let $G$ be a node-labeled graph with dangling edges $D=\{d_1, \dots, d_k\} \subseteq E(G)$.
  We say that graph $G$ \emph{realizes} $R$ if for all $D' \subseteq D$:
  $D' \in R$
  if and only if
  there is a solution $S \subseteq E(G)$ with $S \cap D = D'$.
  We say that $G$ \emph{$B$-realizes} $R$ if $G$ is $B$-homogeneous.
  The endpoints of the dangling edges are called \emph{portals}.
\end{definition}
The crucial part of the reduction is the proof of the following theorem.
We postpone its proof and first show the lower bound.
\begin{theorem}\label{thm:dec:realization}
  Let $B \subseteq \SetN$ be a fixed set of finite size with $\maxgap B >1$ and $0 \notin B$.
  There is a $f:\SetN\to\SetN$ such that the following holds.
  Let $R \subseteq \SetB^e$ be an even relation
  (i.e.\ $\hw(x)$ is even for all $x\in R$).
  Then we can $B$-realize $R$ by a simple graph with $f(e)$ vertices
  of degree at most $\max B+2$,
  the portal nodes are pairwise distinct.
\end{theorem}
Recall the formal statement of the insertion from \cref{def:insertion}.
The following lemma follows directly from this definition and the definition of the realization.
We use it to replace the relations by their realization.
\begin{lemma}\label{lem:dec:insertingIsSafe}
  Let $G$ be a node labeled graph,
  $R$ the relation of node $v\in V(G)$,
  and let $H$ realize $R$.
  Assume inserting $H$ in $G$ at $v$ gives us a new graph $G'$.
  Then there is a solution for $G$ if and only if there is a solution for $G'$.
\end{lemma}
Now we can prove the lower bound under SETH.
We assume that $B \subseteq \SetN$ is a fixed, finite set such that
$0\notin B$ and $\maxgap B>1$.
\begin{proof}[Proof of \cref{thm:lower:dec}]
  \label{proof:thm:lower:dec}
  Let $H$ be a \BFR instance with $n_S$ simple nodes, $n_C$ complex nodes,
  pathwidth $\pw_H$,
  and $\Delta^* = \max_{\text{bag }X}\sum_{v\in V_C\cap X}\deg(v)$.

  Replace every complex node $v$ and its relation $R_v$
  by a $B$-homogeneous graph of size at most $f(\deg(v))$
  according to \cref{thm:dec:realization} to get the graph $G$
  with $n_G\in\O(n_S + n_C \cdot f(\Delta^*))$ vertices.
  Obviously we can bound the pathwidth of the inserted graphs by their size.
  We transform each bag $X$ of the path decomposition of $H$ into a bag $X'$ of the path decomposition of $G$
  by replacing all complex nodes with the nodes of their realization.
  \[
    \abs{X'} \le \abs{X} + \sum_{v \in X\cap V_C} f(\deg(v))
    \le \abs{X} + \sum_{v \in X\cap V_C} f(\Delta^*)
    \le \abs{X} + \Delta^* f(\Delta^*)
    =   \abs{X} + f'(\Delta^*)
  \]
  for some function $f'$.
  Hence,
  the pathwidth of $G$ is bounded by $\pw_G \le \pw_H+f'(\Delta^*)$.

  Now assume we can solve \BFactor in the claimed running time:
  \begin{align*}
    (\max B+1-\epsilon)^{\pw_G} \cdot n_G^{\O(1)} 
    &\le (\max B+1-\epsilon)^{\pw_H+f'(\Delta^*)} \cdot
        (n_S + n_C \cdot f'(\Delta^*))^{\O(1)} \\
    &\le (\max B+1-\epsilon)^{\pw_H+f'(\Delta^*)} \cdot f''(\Delta^*) \cdot (n_S+n_C)^{\O(1)} \\
    &\le (\max B+1-\epsilon)^{\pw_H+f'''(\Delta^*)} \cdot (n_S+n_C)^{\O(1)}
  \end{align*}
  for some $f''$ and $f'''$.
  But this would immediately contradict SETH by \cref{thm:lower:satToBFR}.
\end{proof}

\subsection{Realizing Relations}
From now on
let $B \subseteq \SetN$ be our fixed, finite set
with $\min B \ge 1$ and $\maxgap B = d>1$
such that $[a,a+d+1] \cap B = \{a, a+d+1\}$ for some $a\ge1$.
We first realize three quite basic relations which we use later to realize the more complex relations.
\begin{lemma}\label{lem:dec:all}
  We can $B$-realize each of the relations
  $\HWeq[2]{2}$, $\EQ{d+1}$, and $\EQ{2}$
  by a simple graph
  with $\O(\poly(\max B))$ vertices of degree at most $\max B$.
\end{lemma}
\begin{proof}
  \begin{enumerate}
    \item
    Define a $\min B+1$-clique with new vertices.
    Split an arbitrary edge $(u,v)$
    into two dangling edges $(?,u)$ and $(?,v)$.
    The construction of the clique and the fact that we chose $\min B$ as degree
    forces the two dangling edges to be selected in any solution.

    \item
    We start with two new vertices $u,v$ and connect each to $a$ many common $\HWeq[2]{2}$ nodes.
    We add $d+1$ dangling edges to $u$ and zero to $v$.
    Finally the nodes are replaced by their realization. 
    Observe that $u$ has $a$ forced edges and $d+1$ dangling edges. 
    Thus we must select none or all of the dangling edges since $[a,a+d+1]\cap B=\{a,a+d+1\}$.

    \item
    Define a $d+2$-clique with $\EQ{d+1}$ nodes.
    Split an arbitrary edge $(u,v)$ into two dangling edges $(?,u)$ and $(?,v)$.
    Replace the nodes by their realization.

    Either both dangling edges are selected in which case all nodes have $d + 1$ incident edges in the solution,
    or neither is selected in which case every node has zero incident edges in the solution.
    \qedhere
  \end{enumerate}
\end{proof}
The following lemma helps us to keep the later constructions simple.
Instead of constructing the relations for arbitrary degree, only the very low degree cases are necessary.
\begin{lemma}\label{lem:dec:generalNegation}
  If we can realize $\HWeq[a]{1}$ for $a\in \{1,2,3\}$
  by a simple graph
  with at most $N$ vertices of degree at most $D$,
  then we can realize $\HWeq[k]{1}$ for all $k\ge 1$
  by a simple graph
  using $\O(kN)$ nodes of degree at most $D$.
\end{lemma}
\begin{proof}
  We construct the graph for the realization inductively starting with the basis for $k=1,2,3$.
  See \cref{fig:dec:generalNegation} for an example.
  
  For the inductive step from $k$ to $k+1$ we start with a node $u$ with relation $\HWeq[k]{1}$.
  Connect one dangling edge of $u$ to a new node $v$ with $\HWeq[2]{1}$.
  Connect the other dangling edge of $v$ to a node $w$ with relation $\HWeq[3]{1}$.
  Observe that the final graph has $k+1$ dangling edges.

  Assume one dangling edge of $u$ is selected, then the edge between $u$ and $v$ is not selected
  but the edge from $v$ to $w$ is.
  Hence, no dangling edge of $w$ can be selected.
  The analogue holds if one of the dangling edges of $w$ is selected.
  It cannot be the case that more than one or zero dangling edges are selected,
  as then the relation of one of the three nodes $u$, $v$, or $w$ would not be satisfied.
\end{proof}
\begin{figure}
  \centering
  \includegraphics{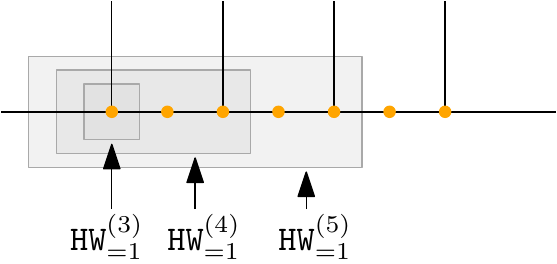}
  \caption{
  Example of the inductive construction from \cref{lem:dec:generalNegation}
  for \HWeq[6]{1} using \HWeq[3]{1} and \HWeq[2]{1}.}
  \label{fig:dec:generalNegation}
\end{figure}
Due to parity issues, the construction of the realizations depends on the set $B$.
For each of the possible cases
($B$ contains only even numbers, only odd number, or even and odd numbers)
we state the result separately in \cref{lem:dec:even,lem:dec:odd,lem:dec:general}.
\begin{lemma}\label{lem:dec:even}
  If $B$ contains only even numbers,
  we can $B$-realize the following relations
  by simple graphs:
  \begin{enumerate}
    \item $\EQ{k}$ for even $k\ge2$ using $\O(k\poly(\max B))$ vertices of degree at most $\max B$.
    \item $\HWeq[k]{1}$ together with $\HWeq[\ell]{1}$ for all $k, \ell\ge 1$ using $\O((k+\ell)\poly(\max B))$ vertices of degree at most $\max B+2$.
  \end{enumerate}
\end{lemma}
\begin{proof}
  \begin{enumerate}
    \item
    For $k=2$ we can use the construction of \cref{lem:dec:all}.
    For the other case we first realize $\EQ{4}$.
    Then we use a chain of these relations to realize $\EQ{k}$ for even $k\ge 6$.

    Start with a $\EQ{d+1}$ node $u$ and make it adjacent to $\frac{d+1-4} 2$ many $\EQ{2}$ nodes (note that an even $B$ can have only gaps of odd size, hence $d$ is odd).
    Then we add four dangling edges to $u$.
    hence the construction actually works.
    The graph is simple as the dangling edges in the realization of $\EQ{2}$ are different.

    \item
    To use \cref{lem:dec:generalNegation} for the general construction,
    observe that the number of $\HWeq{1}$ nodes used in the construction is odd.
    Hence, we will always realize two nodes.
    For this we show how to realize $\HWeq[k]{1}$ together with $\HWeq[\ell]{1}$ for all $k,\ell\in\{1,2,3\}$.

    Start with two vertices $u, v$.
    Make $u$ and $v$ adjacent to $\max B -1$ common $\HWeq[2]{2}$ nodes.
    We add $k$ dangling edges to $u$ and $\ell$ dangling edges to $v$.
    As $B$ does not contain $\max B-1$, the correctness follows.
    \qedhere
  \end{enumerate}
\end{proof}
Despite the fact, that the parity of the elements is still the same
when all elements in $B$ are odd,
it is possible to show a stronger result, which allows us to realize single $\HWeq{1}$ nodes.
\begin{lemma}\label{lem:dec:odd}
  If $B$ contains only odd numbers,
  we can $B$-realize the following relations
  by simple graphs:
  \begin{enumerate}
    \item $\EQ{k}$ for even $k\ge2$ using $\O(k\poly(\max B))$ vertices of degree at most $\max B$.
    \item $\HWeq[k]{1}$ for $k\ge1$ using $\O(k\poly(\max B))$ vertices of degree at most $\max B+2$.
  \end{enumerate}
\end{lemma}
\begin{proof}
  \begin{enumerate}
    \item
    The construction from \cref{lem:dec:even} for \EQ{k} works here too.

    \item
    We again use \cref{lem:dec:generalNegation} and only realize $\HWeq[k]{1}$ for $k=1,2,3$.
    We start with just one node $u$ to which we force $\max B -1$ edges
    by making it adjacent to $\frac{\max B -1}2$ nodes of type $\HWeq[2]{2}$.
    This is possible as $\max B-1$ is even.
    We add $k$ dangling edges to $u$ and replace all nodes by their realizations.
    As before the correctness follows since $B$ does not contain $\max B-1$.
    \qedhere
  \end{enumerate}
\end{proof}
For the last case, where even and odd numbers are in $B$,
we can additionally drop the remaining restriction that the equality nodes must have even degree.
This is rather natural as a gap of size $2$ already gives us a $\EQ{3}$ node by \cref{lem:dec:all}.
\begin{lemma}\label{lem:dec:general}
  If $B$ contains even and odd numbers,
  we can $B$-realize the following relations
  by simple graphs:
  \begin{enumerate}
    \item $\EQ{k}$ for $k\ge1$ using $\O(k\poly(\max B))$ vertices of degree at most $\max B$.
    \item $\HWeq[k]{1}$ for $k\ge1$ using $\O(k\poly(\max B))$ vertices of degree at most $\max B+2$.
  \end{enumerate}
\end{lemma}
\begin{proof}
  \begin{enumerate}
    \item
    The case $k=2$ follows by \cref{lem:dec:all}.

    For the remaining cases it suffices to $B$-realize $\EQ{1}$.
    Then we make $d-2\ge 0$ many of these nodes adjacent to one $\EQ{d+1}$ node.
    Adding three dangling edges and replacing all nodes recursively by their $B$-realization realizes $\EQ{3}$.
    Then we can use a chain of these relations to realize $\EQ{k}$ for all $k\ge 4$.

    Let $b,c \in B$ such that $b<c$ with different parity.
    Force $b$ edges on two new nodes $u,v$ by making both adjacent to $b$ many common nodes of type $\HWeq[2]{2}$.
    Then we make $\frac{c-b-1}2$ many $\EQ{2}$ nodes adjacent to $u$ and add one dangling edge to it.
    The vertex $v$ does not get any dangling edges but is always in a valid state, as all incident edges are selected.
    If the dangling edge incident to $u$ is selected,
    then we also select the edges to the $\EQ{2}$ nodes.
    Otherwise no edge is selected except for the ones forced by the $\HWeq[2]{2}$ nodes.

    \item
    We first realize $\HWeq[1]{1}$ to force single edges.
    For this let $o\in B$ be an odd number.
    Start with a new vertex and force $o-1$ edges to it by making it adjacent to $\frac{o-1}2$ nodes of type $\HWeq[2]{2}$.
    Then we make this vertex adjacent to one more node of that type and add one dangling edge to it.

    We only show the construction for $k=2,3$ and make use of \cref{lem:dec:generalNegation} for larger $k$.
    Start by ``duplicating'' all ingoing edges using $\EQ{3}$ nodes from the first step.
    Force $a+d$ edges on a new vertex $u$ and make it adjacent to one copy of each ingoing edge.
 Vertex   $u$ ensures that at least one of the ingoing edges is selected.
    The other copy of each ingoing edge is connected to a new vertex $v$.
    We force $\max B-1$ edges to $v$.
    By this at most one of the ingoing edges has to be selected.
    \qedhere
  \end{enumerate}
\end{proof}
Now we have everything ready to prove that even relations can be realized.

\begin{figure}
  \centering
  \includegraphics{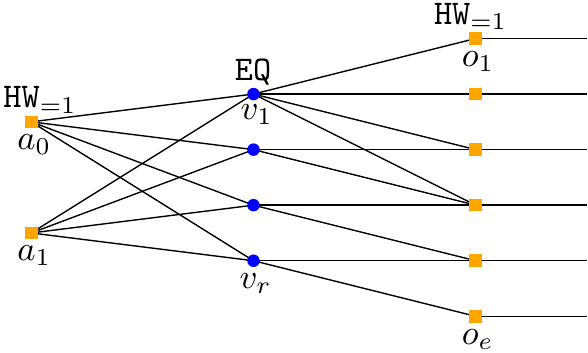}
  \caption{An example illustrating the construction from the proof of \cref{thm:dec:realization}
  for the relation $R$ with
  $R(000011)=R(110011)=R(111001)=R(111100)=1$ and zero otherwise.
  }
  \label{fig:dec:realization:example}
\end{figure}

\begin{proof}[Proof of \cref{thm:dec:realization}]
  See \cref{fig:dec:realization:example} for an example of the following construction.
  We use essentially the construction from Lemma~3.3 in \cite{CurticapeanM16}.
  Let $R = \{x_1, \dots, x_r\} \subseteq \SetB^e$ be the even relation
  for some $r$.
  Let further $P = \{(1+e \mod 2), 0\}$.
  \begin{enumerate}
    \item Create nodes $o_1,\dots,o_e$ with relation $\HWeq{1}$.
    \item Create vertices $a_j$ for all $j \in P$ with relation $\HWeq{1}$.
    \item For all $i \in [r]$:
    \begin{enumerate}
      \item Let $O_i = \{n^{(i)}_1, \dots, n^{(i)}_{h_i}\} = \{ k \in [e] \mid x_i[k] = 0 \}$ for $h_i=e-\hw(x_i)$.
      \item Create the node $v_i$ with relation $\EQ{}$ and connect it to $o_{n^{(i)}_j}$ for all $j \in [h_i]$.
      \item Connect $v_i$ to all $a_j$.
    \end{enumerate}
    \item Replace all nodes by their realization.
  \end{enumerate}
  There are $\abs{P} + e$ many $\HWeq{1}$ nodes.
  Since $\abs{P} = 1 + (1+e \mod 2)$,
  we can replace pairs of these nodes by their realization.
  Every $v_i$ is connected to $\abs{P} + \abs{O_i}$ nodes,
  where $\abs{O_i} = e - \hw(x_i)$.
  Thus, $v_i$ has even degree as the relation $R$ is even,
  i.e.\ $\hw(x)$ is even.
  Hence, we can replace these nodes by their realization according to the previous lemmas.

  To show that the construction actually realizes the relation,
  assume the selected dangling edges corresponds to some element $x \in R$,
  let it w.l.o.g.\ be $x_1$.
  Then we can select all edges incident to $x_1$, the dangling edges, and the extension of this to all nodes as a solution.
  As $x_1$ is incident to all $a_j$ they are in a valid state.
  Further $x_1$ is adjacent to those $o_k$ where $x_1[k]=0$ and hence every $o_k$ is incident to exactly one edge in the solution.

  Now assume we are given a solution.
  As the nodes $a_j$ have exactly one incident edge in the solution, there is exactly one node $v_i$ where all incident edges are in the solution.
  Let $O$ be the set of nodes $o_k$ to which $v_i$ is incident.
  By construction $v_i$ corresponds to some $x \in R$ with $x[k]=0$ iff $k\in O$.
  As all selected dangling edges must be in the solution, let $O'$ be the set of nodes incident to the selected dangling edges.
  But as we are given a solution we get $O \dotcup O' = \{o_1,\dots,o_e\}$.
  Hence, the dangling edges corresponds to $x$.
\end{proof}

\section{Optimization Version}\label{sec:opt}
In the previous section we have seen the realization of the relations for the decision version.
As we are interested in the largest solution for \MaxBFactor,
we also allow $0\in B$
since this does not make the problem trivially solvable.
This makes it necessary to change the definition of a realization,
as the pure existence of a solution is not sufficient anymore.
We change it such that
if the relation is satisfied (i.e.\ the dangling edges are selected in a good way),
then there is a \emph{large} solution.
Otherwise, there must be a gap by which any solution is \emph{smaller}
compared to the solutions in the good cases.
We call this gap the penalty (of the realization).
\begin{definition}[Realization]\label{def:opt:realization}
  Let $R \subseteq \SetB^k$ be a relation.
  Let $G$ be a node labeled graph,
  with dangling edges $D=\{d_1, \dots, d_k\}$.
  We say that graph $G$ \emph{realizes} $R$ with \emph{penalty} $\beta$
  if we can efficiently construct/find
  a target value $\alpha>0$
  such that for all $D' \subseteq D$:
  \begin{itemize}
    \item
    If $D' \in R$,
    then there is a solution $S \subseteq E(G)$ with $S \cap D = D'$ and
    $\abs{S} = \alpha$.

    \item
    If $D' \notin R$,
    then for all solutions $S \subseteq E(G)$ with $S \cap D = D'$
    we have $\abs{S} \le \alpha-\beta$.
  \end{itemize}
  We say that $G$ $B$-realizes $R$ if $G$ is additionally $B$-homogeneous.
  We call the endpoints of the dangling edges \emph{portal nodes}.
\end{definition}
In the main part of this section we show how to realize the relations of \BFR.
The following theorem corresponds to \cref{thm:dec:realization} for the decision version.
\begin{theorem}[Realization of Relations]\label{thm:opt:realization}
  Let $B\subseteq \SetN$ be a fixed, finite set 
  with $\maxgap B > 1$ and $0\in B$.
  There is a $f:\SetN^2\to\SetN$ such that the following holds.
  Let $R \subseteq \SetB^e$ be a relation with a constant $c_R \in 2\SetN$
  such that for all $x \in R$ we have $\hw(x)=c_R$.

  We can $B$-realize the relation $R$ with arbitrary penalty $\beta>0$
  by a simple graph
  with $f(e, \beta)$ vertices of degree at most $\max B+2$.
\end{theorem}
\begin{note*}
  Consider a graph $G$
  where we replaced nodes and their relations by realizations to get a graph $G'$.
  Then we say that a solution $S$ for $G'$ is \emph{valid or good}
  if $S\cap E(G)$ is a solution for $G$.
  Otherwise the solution is \emph{invalid or bad}.
  Consider a node $v$ in $G$.
  Let $I$ be the incident edges
  and $E_v$ be the edges in $G'$ resulting from replacing $v$ by its realization.
  When considering a solution $S$ for $G$,
  we say that $S\cap(I(v) \cup E_v)$ is the partial solution for the realization of $v$,
  i.e.\ all selected edges from the inserted graph and the incident edges.
\end{note*}  

It remains to compute the target value
by which we decide if the \BFR instance has a solution or not.
\begin{lemma}\label{lem:opt:targetValue}
  Let $G$ be a \BFR instance from \cref{sec:lower}.
  Let $G'$ be a \BFactor instance resulting from $G$
  by replacing every complex nodes with degree $\delta$
  by its realization with penalty $2\delta$.
  Then, 
  there is an efficiently computable constant $\alpha$ such that
  $G$ has a solution
  if and only if
  the largest solution for $G'$ has size $\alpha$.
\end{lemma}
\begin{proof}
  By \cref{corr:lower:construction:observations} the complex nodes form a vertex cover for the graph.
  Further note,
  that two complex nodes share at most two edges
  and by our definition of the relations
  exactly one edge from each such pair of edges is selected in a valid solution.
  Let $X$ be the number of these pairs of edges.
  
  Let $\alpha_v$ be the target value for the realization of complex node $v$.
  Then we define \(
    \alpha \deff \sum_{v\in V_c} \alpha_v - X
  \).
  The negative term accounts for the edges between the complex nodes,
  as they are counted twice.
  
  The ``only if'' direction of the claim follows directly by the definition of $\alpha$.
  For the ``if'' direction
  assume $G'$ has a solution of size at least $\alpha$. 
  
  It suffices to show that this solution is valid,
  i.e.\ the relation of all complex nodes are satisfied.
  Then the solution for $G$ follows directly by choosing the edges incident to the realizations of the complex node.
  Assume otherwise and let $W$ be the set of complex nodes whose relation is not satisfied.
  As the relation of each $v\in W$ is not satisfied,
  we lose at least a factor of $\beta_v$.
  But as all incident edges of $v$ could be selected,
  the size of the partial solution for $v$ is at most $\alpha_v-\beta_v+\deg(v)$.
  But as we chose $\beta_v=2\deg(v)$,
  this solution has size at most $\alpha_v-\deg(v)$.
  
  For $i\in\{0,1,2\}$,
  let $Y_i$ be the number of selected \emph{edges} between complex nodes,
  where $i$ endpoints are in $W$.
  Then $Y_0+Y_1+Y_2$ corresponds to the total number of selected edges between complex nodes.
  As these edges are counted twice,
  they appear negative in the following bound for the total size:
  \[
      \sum_{v\in W} \alpha_v - \deg(v)
    + \sum_{v \in V_c\setminus W} \alpha_v
    - Y_0 - Y_1 - Y_2
  \]
  If we can show that this is strictly smaller than $\alpha$,
  we get a contradiction to the assumption that the size of the solution is at least $\alpha$
  and the solution is valid.

  For $i\in\{0,1,2\}$,
  let $X_i$ be the number of \emph{pairs} of edges between complex nodes
  where $i$ endpoints are contained in $W$.
  Observe that $X_0+X_1+X_2=X$.
  Now it suffices to show:
  \[
    - \sum_{x\in W} \deg(v) - Y_0 - Y_1 - Y_2 < - X_0 - X_1 - X_2
  \]
  For each pair of edges contributing to $X_0$,
  there is exactly one edge selected,
  as the relation of both endpoints is satisfied.
  This implies $X_0 = Y_0$.
  Further, for all pairs of edges contributing to $X_1$,
  we know that exactly one edge of this pair is selected.
  Otherwise, the relation of both nodes would not be satisfied and both nodes would be in $W$,
  thus $X_1=Y_1$.
  Hence, it remains to show
  \[
    X_2 < \sum_{x\in W} \deg(v) + Y_2
  \]
  Each pair of edges contributing to $X_2$
  has two endpoints in $W$.
  We could split $X_2$ according to the contribution of each node in $W$.
  As each node in $W$ can definitely contribute only to at most $\deg(v)/2$ pairs (strictly speaking actually only 2),
  this finishes the proof.
\end{proof}
Now we are ready to prove the conditional lower bound for \MaxBFactor
when $0\in B$.
\begin{proof}[Proof of \cref{thm:lower:maximization}]
  Use \cref{lem:opt:targetValue} to construct the final graph and the target value.
  Then the proof goes analogous to the proof for the decision version (cf.\ \cref{thm:lower:dec}).
\end{proof}
\subsection{High Girth Graphs}
Our definition of realizations includes the dangling edges in the solution size.
We also define a slight variation
where the dangling edges are not included.
While both definitions are equivalent for $\HWeq{k}$ nodes,
they differ for $\HWin{S}$ nodes.
The \emph{internal realization} of a node can be seen as a realization,
where we do not adjust the solution size according to the number of selected dangling edges,
i.e.\ only the internal edges are considered for the solution size.
\begin{definition}[Internal Realization]\label{def:opt:simulation}
  Let $R \subseteq \SetB^k$ be a relation.
  Let $G$ be a node labeled graph,
  with dangling edges $D=\{d_1, \dots, d_k\} \subseteq E(G)$.
  $G$ \emph{internally realizes} $R$ with \emph{penalty} $\beta$
  if we can efficiently construct/find
  an $\alpha>0$
  such that for all $D' \subseteq D$:
  \begin{itemize}
    \item
    If $D' \in R$,
    then there is a solution $S \subseteq E(G)$ with $S \cap D = D'$
    and $\abs{S \setminus D} = \alpha$.

    \item
    If $D' \notin R$,
    then for all solutions $S \subseteq E(G)$ with $S \cap D = D'$
    we have $\abs{S \setminus D} \le \alpha-\beta$.
  \end{itemize}
  We say that $G$ internally $B$-realizes $R$ if $G$ is additionally $B$-homogeneous.
\end{definition}
We know that there is a gap of size at least two between $a$ and $a+d+1$ in $B$.
This allows us to define relatively simple conditions of the form
``if one incident edge of a vertex with degree $a+d+1$ is not selected,
then another edge is also not selected''.
In other words, this propagates the penalty to a neighboring vertex.
We combine this with high girth graphs to introduce an arbitrary large penalty for not selecting an edge.

The construction of $r$ regular graphs with girth $g$ is a long studied problem in graph theory.
Erd\H{o}s and Sachs proved the existence of such graphs
for all combinations of $r$ and $g$.
\begin{lemma}[Theorem~1 in \cite{ErdosS63}]
  \label{cor:opt:highGirth}
  For all $r \ge 2$ and $g\ge 3$,
  there is a $r$-regular graph $G_{r,g}$
  of girth $g$ with at most $4gr^g$ vertices.
\end{lemma}
Finding the smallest graph for each $r,g$ is a non-trivial task and known as the $(r,g)$-cage problem.
For several cases (e.g.\ $r$ is a prime power)
constructions are known reducing the number of vertices in the graph.
See \cite{Dahan14,ExooJ12,Imrich84,LazebnikUW95} for more results.

\subsection{Realizing Relations}
From now on
let $d \deff \maxgap B > 1$ such that $[a,a+d+1] \cap B = \{a, a+d+1\}$ for some $a\ge0$.
As we allow $0\in B$, we can always find a trivial solution.
Thus, we cannot force edges as we did for the decision version.
Instead we construct a gadget
where we can select many edges when the ``forced'' edges are selected.
Otherwise we ensure that the solution is small.
We use the graphs with high girth for this.

\begin{lemma}\label{lem:opt:all}
  There is a $f:\SetN\to\SetN$ such that the following holds.
  We can
  $B$-realize $\HWeq[2]{2}$
  (with distinct portal vertices)
  and internally $B$-realize $\EQ{d+1}$
  with arbitrary penalty $\beta$
  by simple graphs
  using at most $f(\beta)$ vertices
  of degree at most $\max B$.
\end{lemma}
\begin{proof}
  \begin{enumerate}
    \item
    For ${\HWeq[2]{2}}$ we use \cref{cor:opt:highGirth}
    to get an $a+d+1$-regular graph $G_{a+d+1,\beta}$ of girth at least $\beta$.
    Split an arbitrary edge $(u,v)$
    into two dangling edges for $u$ and $v$ each and assign the set $B$ to every vertex.

    The graph has the claimed properties:
    If both dangling edges are selected, then we can use the set of all edges in the graph as a solution since $a+d+1\in B$.

    It remains to check the case when at least one dangling edge is not selected,
    let it w.l.o.g.\ be the one incident to $u$.
    Assume $S$ is the optimal solution.
    We show that this solution does not contain more than $\abs{E_{a+d+1,\beta}}-\beta$ edges.

    By assumption $\deg_S(u) \le a$.
    Hence, there must be at least one other incident edge to $u$ that is not in the solution, because $a+d-1 \ge a+1 \notin B$.
    Then we can apply this argument always to the next vertex.
    Observe that this sequence can only stop if we reach another vertex $w$ we have already visited
    because for this vertex we already know that two incident edges were not selected in the solution.
    The length of this path, i.e.\ the number of not selected edges from $w$ to $w$, is at least the girth of the graph.
    Hence the number of edges that are not selected in the solution
    is at least the girth of the graph which is at least $\beta$.

    \item
    For $\EQ{d+1}$ we start with two vertices $u$ and $v$
    and make both adjacent to $a$ many common $\HWeq[2]{2}$ nodes with penalty $\beta+2$.
    Add $d+1$ dangling edges to one of the vertices and zero to the other.
    As the vertices connected to the dangling edges in the realization of the $\HWeq[2]{2}$ nodes are different, the graph is simple.

    If $d+1$ or $0$ dangling edges are selected,
    we can choose all internal edges for the optimal solution $S$,
    as $a + d+1,a \in B$.
    The size of the solution is in both cases optimal.

    Now assume $0<\ell<d+1$ dangling edges are selected.
    From $a+\ell \notin B$ it follows that
    there must be at least one $\HWeq[2]{2}$ node where at least one incident edge is not selected.
    By this we lose at least $\beta+1$ compared to the optimal solution.
      \qedhere
  \end{enumerate}
\end{proof}
As for the decision version,
we now realize relations used for the general construction.
\begin{lemma}\label{lem:opt:all:negation}
  There is a $f:\SetN^2\to\SetN$ such that the following holds.
  We can $B$-realize a pair of $\HWeq[k]{1}$ and $\HWeq[\ell]{1}$ for any $k, \ell\ge 1$
  with arbitrary penalty $\beta$
  with simple graphs
  using at most $f(k+\ell, \beta)$ vertices
  of degree at most $\max B+2$.
\end{lemma}
For the proof of the lemma we reuse the construction we have already seen for the decision version,
by which an explicit construction of these nodes for $k,\ell \in \{1,2,3\}$ is sufficient.
To realize these relations we distinguish between the cases $\max B-1 \in B$ or not.
By this the proof of \cref{lem:opt:all:negation} follows from
\cref{lem:opt:generalNegation,lem:opt:negation:notInTheSet,lem:opt:negation:inTheSet}.
\begin{lemma}\label{lem:opt:generalNegation}
  If we can $B$-realize $\HWeq[a]{1}$ together with $\HWeq[b]{1}$ for $a,b\in \{1,2,3\}$,
  with arbitrary penalty $\beta$
  using simple graphs
  with at most $N$ vertices of degree at most $D$,
  then we can $B$-realize $\HWeq[k]{1}$ together with $\HWeq[\ell]{1}$ for all $k,\ell\ge 1$
  with arbitrary penalty $\beta$
  using simple graphs
  using $\O((k+\ell)N)$ vertices of degree at most $D$.
\end{lemma}
\begin{proof}
  We use the same construction as for the decision version as stated in \cref{lem:dec:generalNegation}.
  We use $\beta+3$ as penalty for the realizations.
  The target value is the sum of the $\alpha$s of all involved nodes
  minus the number of $\HWeq[2]{1}$ nodes used to connect the other nodes.
  This negative term takes care that we do not count edges twice.

  For the correctness we observe that there could be solutions that do not satisfy the relations of all nodes.
  Note that each node has degree at most 3 and is incident to exactly one selected edge for a valid solution.
  Hence, for every such pair of additionally selected edges,
  there is at least one node not in a valid state.
  By its penalty of $\beta+3$, it can compensate for these additional edges.
  Hence, these solutions are smaller by at least $\beta$.
\end{proof}
For the case $\max B-1 \notin B$,
we use the same approach as for the decision version
when all elements of $B$ are even.
\begin{lemma}\label{lem:opt:negation:notInTheSet}
  There is a $f:\SetN\to\SetN$ such that the following holds.
  If $\max B -1 \notin B$,
  we can $B$-realize a pair of $\HWeq[k]{1}$ and $\HWeq[\ell]{1}$ for $k,\ell\in \{1,2,3\}$
  with arbitrary penalty $\beta$
  with simple graphs
  using at most $f(\beta)$ vertices
  of degree at most $\max B+2$.
\end{lemma}
\begin{proof}
  Start with two vertices $u, v$.
  Connect both to $\max B-1$ common nodes of type $\HWeq[2]{2}$.
  Add $k$ dangling edges to $u$ and $\ell$ dangling edges to $v$.
  Replace all nodes by their realization with penalty $\beta+k+\ell$.
  The target value $\alpha$ is the sum of target values of the $\HWeq[2]{2}$
  plus 2 for two selected dangling edges.

  Assume exactly one dangling edge to each $u$ and $v$ is selected.
  Then we can extend this to a solution for the $\HWeq[2]{2}$ nodes, by selecting all internal edges.

  If for $u$ or $v$ no dangling edge is selected,
  then there must be at least one $\HWeq[2]{2}$ node which is not incident to two edges in the optimal solution
  as $\max B-1 \notin B$.
  Hence, we lose a penalty of $\beta+k+\ell\ge\beta$.
  Now let $\hat k, \hat \ell$ dangling edges of $u,v$ be selected, respectively,
  such that $\hat k>1$ or $\hat \ell >1$.
  W.l.o.g.\ assume $\hat k>1$.
  Then at least one edge between $u$ and a $\HWeq[2]{2}$ node is not selected,
  as $\max B -1 + \hat k>\max B$.
  Hence, the $\HWeq[2]{2}$ node is not incident to two edges
  and we lose a factor of $\beta+k+\ell$.
  As $\hat k+\hat \ell \le k+\ell$,
  we still lose $\beta$ compared to the optimal solution in the good case.
\end{proof}
For the case $\max B-1 \in B$,
we use an approach that is similar to the decision version,
when $B$ contains even and odd numbers.
But as we cannot use an $\EQ{3}$ node,
the construction gets slightly more complicated
as we have to handle the additional edges of a $\EQ{d+1}$ node.
\begin{lemma}\label{lem:opt:negation:inTheSet}
  There is a $f:\SetN\to\SetN$ such that the following holds.
  If $\max B -1 \in B$,
  we can $B$-realize a node $\HWeq[k]{1}$ for $k\in\{1,2,3\}$
  with arbitrary penalty $\beta$
  with simple graphs
  using at most $f(\beta)$ vertices
  of degree at most $\max B+2$.
\end{lemma}
\begin{proof}
  We first prove the claim for $k=1$, which we need for the proof of the other cases.
  \begin{claim}
  There is a $f':\SetN\to\SetN$ such that 
  we can $B$-realize $\HWeq[1]{1}$
  with arbitrary penalty $\beta$
  with simple graphs
  using at most $f'(\beta)$ vertices
  of degree at most $\max B$.
  \end{claim}
  \begin{claimproof}
  We use the same approach as for the decision version
  and let $v$ be a new vertex.
  Let $o$ be the odd number of $\max B-1$ and $\max B$.
  We force $o-1$ edges to the vertex by adding $\frac{o-1}2$ $\HWeq[2]{2}$ nodes to it.
  We introduce one more $\HWeq[2]{2}$ node $u$.
  We add one edge between $u$ and $v$ and add one dangling edge to $u$.
  We replace all nodes by their realization according to \cref{lem:opt:all}
  with penalty $\beta+1$.
  The target value is the sum of all target values of the $\HWeq[2]{2}$ nodes.

  If the dangling edge is selected, then we can select all other edges.
  Vertex $v$ is then incident to $2 \cdot \frac{o-1}2+1 = o \in B$ edges.
  But if the dangling edge is not selected, the node $u$ cannot be in a valid state and we lose at least $\beta$ compared to the optimal solution.
  \end{claimproof}

  Now we prove the construction for $k=2, 3$.
  See \cref{fig:opt:negation:inTheSet} for an example of the construction.
  We create vertices $u_1,\dots,u_{d-1}$ and $v$.
  We force $\max B-1$ edges on $u_j$ for all $j\in[d-1]$
  and $a+d$ edges on $v$ with nodes from the previous claim.
  We further introduce nodes $s_1, \dots, s_k$ with relation $\EQ{d+1}$.
  We make each of the $s_i$ adjacent to the vertices $u_1,\dots,u_{d-1},v$.
  Finally, we add one dangling edge to each $s_i$.
  The vertices $u_j$ ensure that \emph{at most one} dangling edge is selected while $v$ makes sure that \emph{at least one} is selected.
  We replace all nodes by their realization with a penalty of $\beta+3$.
  The target values is the sum of all target values of the realizations for the nodes,
  plus $d+1$ accounting for the selected edges, when one dangling edge is selected. 
  Observe that the edges forced by the $\HWeq[1]{1}$ nodes are usually part of any solution.
  Hence, we only have to focus on the edges incident to the nodes with $\EQ{d+1}$,
  as they are only internally realized.
  Note, this simplification can lead to solutions of negative size,
  when considering the penalty.

  Assume only the dangling edge of $s_i$ is selected.
  Then we select the edges incident to $s_i$ as a solution of size $d+1$.

  If no dangling edge is selected, then the vertex $v$ cannot be in a valid state
  or one of the $s_i$ is not in a valid state.
  This bounds the solution size by $d-(\beta+3) \le d+1-\beta$.

  Now assume $\ell >1$ dangling edges are selected.
  W.l.o.g.\ we can assume that none of the $s_i$ is in an invalid state,
  as this would lead to a solution smaller than for some appropriate $\ell'$.
  As each of the $d-1$ vertices $u_j$ cannot be incident to more than $\max B$ edges in any solution,
  for each selected dangling edge except the first one
  there must be an edge to an adjacent $\HWeq[1]{1}$ node that is not in the solution,
  i.e.\ the $\HWeq[1]{1}$ node is not in a valid state.
  This bounds the solution size by:
  \(
    \ell (d+1) - (d-1)(\ell-1)(\beta+3)
  \).
  If we can bound this by $d+1-\beta$,
  the claim follows.
  But this is equivalent to show $(\ell-1)(d+1) + \beta \le (d-1)(\ell-1)(\beta+3)$.
  As $\ell>1$ and $d>1$, it is easy to check that this claim holds true.
\end{proof}
\begin{figure}
  \begin{subfigure}[t]{0.475\textwidth}
    \centering
    \includegraphics[]{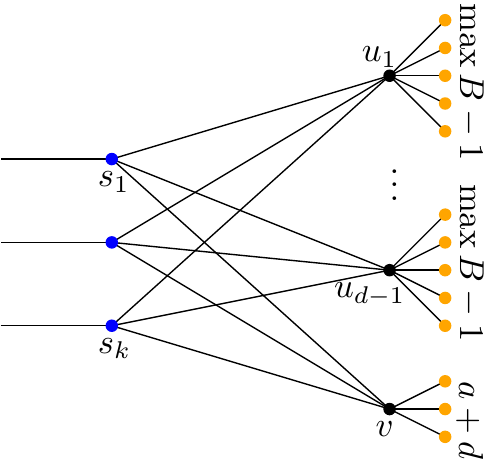}
    \caption{Example construction for the proof of \cref{lem:opt:negation:inTheSet}.}
    \label{fig:opt:negation:inTheSet}
  \end{subfigure}\hfill
  \begin{subfigure}[t]{0.475\textwidth}
    \centering
    \includegraphics[]{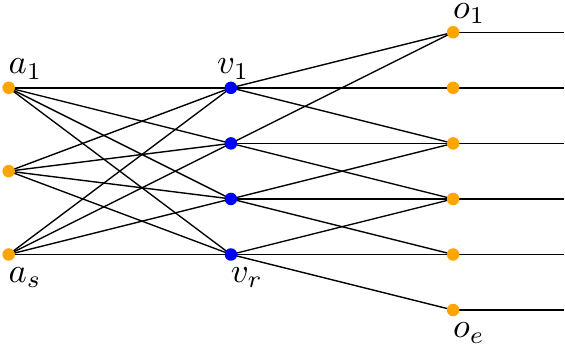}
    \caption{Graph for the relation $R$ with $R(000111)=R(010011)=R(110001)=R(111000)=1$ and zero otherwise.
    }
    \label{fig:opt:all:example}
  \end{subfigure}
  \caption{
  Blue nodes have $\EQ{k(d+1)}$ as relation,
  light blue nodes have $\EQ{1}$,
  black ones \HWin{B},
  and orange nodes have $\HWeq{1}$.
  }
\end{figure}
Following the approach from the decision version,
it remains to realize $\EQ{k}$ nodes for even $k$.
But the construction for the decision version cannot be transferred directly
as it does not give us a realization.
Recall, that we use a $d+2$-clique of $\EQ{d+1}$ nodes
and split one edge into two dangling edges to realize $\EQ{2}$.
But if both dangling edges are selected,
the solution has $\Theta(d^2)$ more edges compared to the solution when no dangling edge is selected.
But surprisingly for $\EQ{k(d+1)}$ nodes a realization can be constructed.
So we use them for our construction.
\begin{lemma}\label{lem:opt:all:equalities}
  There is a function $f:\SetN^2\to\SetN$ such that the following holds.
  For arbitrary $k \ge 1$ we can
  $B$-realize $\EQ{k(d+1)}$
  with arbitrary penalty $\beta$
  by simple graphs
  using $f(k, \beta)$ vertices
  of degree at most $\max B+2$.
\end{lemma}
\begin{proof}
  Construct a $d$-regular connected bipartite graph with $k(d+1)$ nodes of type $\EQ{d+1}$ on each side.
  Let $L$ be the set of nodes on the left side and $R$ of the right side.
  Subdivide each edge by placing one $\HWeq[2]{1}$ node on each edge.
  Add one dangling edge to each node in $L$.
  On the right side we introduce $k$ new nodes of type $\EQ{d+1}$.
  Make each node in $R$ adjacent to exactly one of these nodes such that every node has degree $d+1$.
  Finally, replace all nodes by their (internal) realization that are guaranteed to exist as $d(d+1)$ is always even.
  Use $\beta + 2k(d+1)^2$ as penalty for all nodes.
  The target value is the sum of all target values of the realizations
  plus $k(d+1)$.

  If all $k(d+1)$ dangling edges are selected,
  we can select all $k(d+1)^2$ edges on the left side of the graph
  and none of the edges on the right side as optimal solution.
  The converse holds for the case when no dangling edges are selected.
  No solution can select more edges than this,
  as each additionally selected edge,
  leads to one additional node in an invalid state,
  whose penalty can compensate that edge.

  Now assume $0<\ell<k(d+1)$ dangling edges are selected for some solution.
  Assume for contradictions sake that all nodes in $L$ are in a valid state.
  But then there is a node in $L$ that is in the 0-state and one node is in the $d+1$-state.
  As our bipartite graph is connected,
  there are two such nodes that are additionally connected to the same node in $R$ (by two subdivided edges).
  This implies that this node on the right side is not in a valid state.
  Hence, its penalty can even compensate the selection of all edges.
\end{proof}
Now we have everything together to realize almost arbitrary relations and prove the main theorem for the maximization version.
See \cref{fig:opt:all:example} for an example of the construction. 
\begin{proof}[Proof of \cref{thm:opt:realization}]
  We use a slight variation of our construction for the decision version in \cref{thm:dec:realization}
  and hence the construction in \cite{CurticapeanM16}
  as we cannot realize $\EQ{k}$ for arbitrary even $k$.
  We resolve this issue by using nodes with larger degree
  and add nodes that take care of the additional edges.

  Let $R = \{x_1, \dots, x_r\} \subseteq \SetB^e$ be a relation,
  $c_R \in 2\SetN$ be the constant such that
  for all $x \in R$ we have $\hw(x)=c_R$,
  and let $h \deff e - c_R$
  be the number of 0s of each $x_i$.
  There is a unique integer $\ell>0$ such that $2(\ell-1)(d+1) < h+1 \le 2\ell(d+1)$.
  Let $s \deff 2\ell(d+1)-h$.
  We define the graph realizing $R$ as follows:
  \begin{enumerate}
    \item Create nodes $o_k$ for all $k\in [e]$ and assign the relation $\HWeq{1}$ to them.
    Their degree is determined later, and by \cref{lem:opt:all:negation} not relevant for us.

    \item Create nodes $a_j$ for all $j \in [s]$ with relation $\HWeq{1}$.

    \item For all $i \in [r]$:
    \begin{enumerate}
      \item Create a node $v_i$ with relation $\EQ{2\ell(d+1)}$.
      \item Make $v_i$ adjacent to $a_j$ for all $j \in [s]$.
      \item Let $O_i = \{n^{(i)}_1, \dots, n^{(i)}_{h}\} = \{ k \in [e] \mid x_i[k] = 0 \}$
      and connect $v_i$ to $o_{n^{(i)}_j}$ for all $j \in [h]$.
    \end{enumerate}
  \end{enumerate}
  To get a $B$-homogeneous graph, we replace the nodes $v_i$ by the realization from \cref{lem:opt:all:equalities} which is possible since $h+s = 2\ell(d+1)$.
  All $\HWeq{1}$ nodes, i.e.\ $a_i$ and $o_k$, are replaced by their realization from \cref{lem:opt:all:negation}.
  This is possible as $s+e = 2\ell(d+1)-h+e = 2\ell(d+1)+c_R$ is an even number.
  For all realizations we define the penalty to be $\beta$.
  Observe that we can see the nodes with $\HWeq{1}$ as internal realizations.
  That is the target value for their realization does not include the incident edges in the count.
  Then the target value is the sum of all target values of the realizations
  plus $c_R$ to account for the dangling edges.

  When the selected dangling edges correspond to some element $x_i \in R$,
  we select all edges incident to $v_i$
  and the extension to the realizations of the nodes.
  Then all nodes are in a valid state.
  There is no larger solution, as then one of the $\HWeq{1}$ nodes would be in an invalid state.
  Further, for the $\EQ{2\ell(d+1)}$ nodes it makes no difference
  if all incident edges are selected or not, as they are realized.
  As the number of selected dangling edges is always the same,
  the solution size is always the same.

  Now assume we are given a solution, where the dangling edges do not corresponds to some $x\in R$.
  We show that this solution must be smaller by at least $\beta$ compared to the optimal solution size.
  We can assume that all $v_i$ nodes are in a valid state.
  Otherwise we would directly lose a factor of $\beta$ for each of these node in an invalid state,
  as these are realized.
  Assume there is no $v_i$ where all dangling edges are selected.
  Then all nodes $a_j$ are in an invalid state,
  as each $a_j$ is connected to all $v_i$.
  Hence, we lose a factor of at least $\beta$.
  If there is more than one $v_i$ where all incident edges are selected,
  then by the same argument all $a_j$ are in an invalid state
  as they are connected to at least two edges
  and we lose $\beta$.
  Hence, there must be exactly one $v_i$ where all incident edges are selected.
  Then all nodes $a_j$ are in a valid state.
  Let $O'$ be the set of the remaining nodes adjacent to $v_i$ in the solution.
  By design, this corresponds to a $x\in R$ with $x[k]=0$ iff $k\in O'$.
  Let $O''$ be the nodes incident to a dangling edge in the solution.
  Assume there is a $o \in O' \cap O''$.
  Then this node would be incident to two edges and we lose $\beta$.
  The same argument holds for nodes $o \notin O'\cup O''$.
  Hence, $O'$ is the complement of $O''$ and the dangling edges correspond to an element of $R$,
  a contradiction.
  Thus, some node must be in an invalid state
  and the solution is smaller by at least $\beta$.
\end{proof}

\section{Counting Version}\label{sec:count}
From a certain perspective the optimization version can be seen as a relaxation of the decision version:
The assumption $\min B >0$ is dropped
while still assuming $\maxgap B >1$.
For the counting version we now even drop this last assumption
such that there might be no gap at all in $B$.
Thus the only polynomial-time solvable cases for the counting version are $B=\{0\}$ and $B=\emptyset$ with one and zero solutions, respectively.
This implies that we additionally must realize equality relations.
Surprisingly this also reduces to realizing \HWeq[1]{1} nodes in the end,
i.e.\ forcing edges.

We use the Holant framework and lemmas and definitions analogous to those from \cite{CurticapeanM16}.
A signature graph $\Omega$ is a graph with weights $w_e$ for all edges $e$
and all vertices are labeled by \emph{signatures} $f_v:\SetB^{I(v)} \to \mathbb{Q}$,
which are rational functions on the incidence vector $I(v)$ of the edges adjacent to $v$.
We define $\Hol{\Omega}$ to be the quantity
\begin{align*}
	\sum_{x\in \{0,1\}^{E(\Omega)}} \prod_{e\in x} w_e \prod_{v\in V(\Omega)} f_v(x|_{I(v)}).
\end{align*}
The Holant framework can be seen as a natural generalization of \GenFac.
If each signature $f_v$ is a symmetric Boolean function
and each edge weight is 1,
then it is exactly \CountGenFac.
If additionally each vertex has signature $\HWin{B}$,
this corresponds to \CountBFactor.
\begin{definition}[$\hol{F}$]
	If $F$ is a set of rational functions,
	we say that $\hol{F}$ is the set of all Holant problems
	where the signature graph has signatures only from $F$.
\end{definition}
\begin{definition}[Gate]
	A gate is a signature graph $\Gamma$, possibly containing a set $D\subseteq E(\Gamma)$ of dangling edges, all of which have edge weight 1. The signature realized by $\Gamma$ is the function $\Sig(\Gamma):\SetB^D\to \mathbb{Q}$ that maps an assignment of dangling edges $x\in \SetB^{D}$ to 
	\begin{align*}
	\Sig(\Gamma,x) =
	\sum_{y\in \SetB^{E(\Gamma)\setminus D}}
		\left(
			\prod_{e\in E(\Gamma)} w(e)
			\prod_{v\in  V(\Gamma)} f_v\left ((x\cup y)|_{I(v)}\right )
		\right)
	\end{align*}
\end{definition}
	Note that unless mentioned otherwise,
	we restrict ourselves to signature graphs with unit edge weights and hence they are usually omitted.

In essence, gates in the Holant framework play the role of realizations in the previous sections. 
Given these definitions, we are now ready to state our main theorem,
which can then be used to prove \cref{thm:lower:counting}.
\begin{theorem} \label{thm:count:bfr-to-gen-matching}
	For all fixed, finite $B\subseteq \SetN$ with $B\neq\{0\}$
	there is a $f:\SetN\to\SetN$ such that the following holds.
	Let $G=(V_S\dotcup V_C, E)$ be an instance of \CountBFR
	with a path decomposition of width $\pw$
	such that $\Delta^* = \max_{\text{bag } X} \sum_{v \in X\cap V_C} \deg(v)$.
	Then there is a $\Ostar{f(\Delta^*)}$ time Turing reduction from \CountBFR to \CountBFactor
	such that for every constructed instance of \CountBFactor pathwidth and cutwidth increase at most by $f(\Delta^*)$.
\end{theorem}
We postpone the proof of this reduction
and first show the lower bound for the counting version.
\begin{proof}[Proof of \cref{thm:lower:counting}]
	Let $H$ be a \CountBFR instance with $n_H$ nodes,
	pathwidth $\pw_H$,
	and $\Delta^* = \max_{\text{bag }X}\sum_{v\in V_C\cap X}\deg(v)$.
	Then by \cref{thm:count:bfr-to-gen-matching},
	we get polynomially many instances of \CountBFactor
	such that $n_G\in (n_H+\Delta^*+f(\Delta^*))^{\O(1)}$
	with the $f$ from the above theorem
	where $\pw_G \leq \pw_H +f(\Delta^*)$.
	Now, suppose that we can solve \CountBFactor in the claimed running time.
	Then we can solve \CountBFR in time
	\begin{align*}
		 (\max B+1-\epsilon)^{\pw_G} \cdot n_G^{\O(1)} 
		&\le (\max B+1-\epsilon)^{\pw_H+f(\Delta^*)}
		\cdot (n_H+f(\Delta^*))^{\O(1)} \\
		&\le (\max B+1-\epsilon)^{\pw_H+f(\Delta^*)}
		\cdot f'(\Delta^*)
		\cdot n_H^{\O(1)} \\
		&\le (\max B+1-\epsilon)^{\pw_H+f''(\Delta^*)}
		\cdot n_H^{\O(1)}
	\end{align*}
	for some $f'$ and $f''$.
	But this immediately contradicts \#SETH by \cref{corr:lower:construction:parsimonious}.
\end{proof}
We can think of \CountBFR as a Holant problem
where the allowed signatures are either $\HWin{B}$ or restricted even relations.
We first use a lemma from \cite{CurticapeanM16} to realize these relations through nodes with signature $\HWeq{1}$.
Since their constructions are in the perfect matching setting,
they can equivalently be seen as gates that use vertices with signature $\HWeq{1}$.
After using this lemma to reduce from \CountBFR to a Holant problem,
we give a chain of reductions (see \cref{fig:count:chain-of-reductions})
that ends at \CountBFactor and preserves the pathwidth up to an additive constant. 

\begin{figure}[t]
	\centering
	\includegraphics{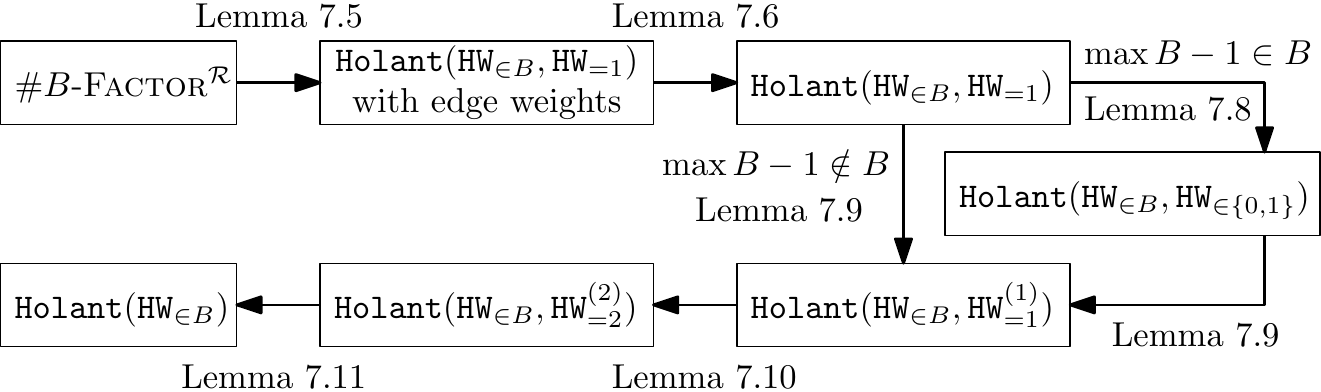}
	\caption{The chain of reductions that starts with \CountBFR
	and ends at \CountBFactor (i.e.\ $\hol{\HWin{B}}$).
	Arrows show the direction of Turing or many-one reductions.}
	\label{fig:count:chain-of-reductions}
\end{figure}

\begin{lemma}[Lemma~3.3 from \cite{CurticapeanM16}]\label{lem:count:simulate-relations}
	Let $R\subseteq \SetB^{e}$ be an even relation.
	Then there is a gate $\Gamma$ that can realize it such that
	\begin{itemize}
		\item $\Gamma$ only uses vertices with signature $\HWeq{1}$.
		\item $\Gamma$ has $\O(\abs{R}\cdot e)$ vertices and edges.
		\item $\Gamma$ has maximum degree at most $\abs{R}+\O(1)$.
		\item the edges of $\Gamma$ have weights from $\{-1,\frac{1}{2}, 1\}$.
		\item given $R$ as input, we can construct $\Gamma$ in time $\O(\abs{R}\cdot e)$.
	\end{itemize} 
\end{lemma}

\begin{lemma}\label{lem:count:BFR-to-hol-edge-weights}
	\CountBFR reduces in polynomial time to $\hol{\HWin{B},\HWeq{1}}$
	with edge weights from $\{-1,\frac{1}{2}, 1\}$
	such that the pathwidth and cutwidth increase at most by a constant $c_{\Delta^*}$ depending only on $\Delta^*$,
	where ${\Delta^*}$ is the maximum total degree of the complex nodes in any bag of the path decomposition,
	and the maximum degree is unaffected, or increases to 6.
\end{lemma}

\begin{proof}
We replace the complex nodes with the appropriate gates from \cref{lem:count:simulate-relations}.
This gives us an instance of $\hol{\HWin{B},\HWeq{1}}$ with edge weights from $\{-1,\frac{1}{2}, 1\}$.
Further, this only alters the pathwidth and cutwidth by an additive factor depending only on $\Delta^*$.
However, the construction via \cref{lem:count:simulate-relations} can result in large degree $\HWeq{1}$ nodes. 
These can be replaced by smaller degree $\HWeq{1}$ nodes
(see the construction for the decision version in \cref{fig:dec:generalNegation}).
\end{proof}
Since \CountBFactor does not have edge weights,
removing them is the next step in our chain of reductions.
We do this through polynomial interpolation. 
This technique was first used by Valiant \cite{Valiant79}. 
We will also use this idea further down our chain of reductions, 
and give a general statement \cref{prop:count:distinct-ratio} to aid us in this process. 
The idea is that we can recover the coefficients of a polynomial $P(\cdot)$ 
if we know the value of $P(x)$ for sufficiently many $x$. 
We represent the solution of one problem as the value of a polynomial $P(\cdot)$ and 
represent the value of a second problem as some function $f(P)$ of the polynomial itself. 
Then, we can recover the value of the second problem through 
sufficiently many invocations to an oracle of the first problem. 
This gives us a Turing reduction from the second problem to the first. 

\begin{lemma}\label{lem:count:hol-edge-weights-to-hol1}
	There is a polynomial time Turing reduction from $\hol{\HWin{B},\HWeq{1}}$ with edge weights from $\{-1,\frac{1}{2}, 1\}$
	to $\hol{\HWin{B},\HWeq{1}}$ on unweighted graphs such that the
	pathwidth and cutwidth increase only by a fixed constant and the maximum degree increases to at least three.
\end{lemma}

\begin{proof}
	We follow the arguments from Theorem~4.1 of \cite{CurticapeanM16}.
	We replace $-1,\frac{1}{2}$ weight edges with weights $x,y$
	and treat \Hol{\Omega} as a polynomial in $x,y$.
	Then, we can recover this polynomial by interpolating with various values of $x,y$.
	We can choose the values for $x,y$ from $\{2^i \mid i\in \SetN \}$ through the constructions from \cite{CurticapeanM16}:
	We replace the edges with weight $2^i$ by a path with $2i+3$ edges
	where edges number $3,5,\dots,2i+1$ have weight $2$ and the other edges have weight $1$.
	Then we replace these edges with weight $2$ by two parallel edges
	and divide both edges by two nodes.
	All new nodes are labeled with $\HWeq{1}$.
	Observe that this construction increases pathwidth and cutwidth only by a constant factor
	and the degree increases to at least $3$.
\end{proof}

	\begin{proposition}\label{prop:count:distinct-ratio}
	Suppose we have two non-zero sequences $\{A_n\}_{n\in \SetN}, \{B_n\}_{n\in \SetN}$ that are related as
	\begin{align*}
		\begin{bmatrix}
			A_n\\
			B_n
		\end{bmatrix}
		= M \begin{bmatrix}
			A_{n-1}\\
			B_{n-1}
		\end{bmatrix}
		= M^{n} U \qquad \text{, where }U= \begin{bmatrix}
			A_{0}\\
			B_{0}
		\end{bmatrix}
	\end{align*}
	and $M$ is a symmetric and invertible $2\times 2$ matrix such that
	$U$ is not an eigenvector of $M$.
	Then $\{\frac{B_n}{A_n}\}_{n\in \SetN}$ is a sequence which does not contain any repetitions. 
\end{proposition}

\begin{proof}
	Suppose not. Then we have $\frac{B_n}{A_n} = \frac{B_{n+r}}{A_{n+r}}$ for some positive integers $n,r$. Then since $\{A_n\},\{B_n\}$, are both non-zero sequences we have
	\begin{align*}
		k\begin{bmatrix}
			A_{n}\\
			B_{n}
		\end{bmatrix} = 
		\begin{bmatrix}
			A_{n+r}\\
			B_{n+r}
		\end{bmatrix}
	\end{align*}
	for some $k$. Then we have
	\begin{align*}
		k\begin{bmatrix}
			A_{n}\\
			B_{n}
		\end{bmatrix} &= kM^n U = M^{n+r}U =
		\begin{bmatrix}
			A_{n+r}\\
			B_{n+r}
		\end{bmatrix}
		\qquad
		\overset{M \text{ invertible}}{\implies}
		\qquad kU = M^r U
	\end{align*}
	implying that $U$ is an eigenvector of $M^r$. Now, since $M$ is symmetric, it has two eigenvectors that are linearly independent. Since these are also eigenvectors of $M^r$ and they are linearly independent, these are the only eigenvectors of $M^r$. Thus, $U$ must also be an eigenvector of $M$, giving a contradiction.
\end{proof}
We will see why we need the following lemma in \cref{lem:count:hol-to-genfactor-forced}.
Observe that a $\HWeq{1}$ node can be thought of as a vertex with list $\{1\}$.
To see the problem, consider the case when the list does not have gaps, for example like in $\{0,1,2,3,4\}$.
Then it is not immediate how to get a node whose list has 1 but does not have 0.
We avoid this by doing interpolation.

\begin{lemma}\label{lem:count:hol1-to-hol01}
	There is a polynomial time Turing reduction from $\hol{\HWin{B},\HWeq{1}}$ to $\hol{\HWin{B},\HWin{\{0,1\}}}$
	such that pathwidth and cutwidth increase at most by 1.
\end{lemma}
\begin{proof}
	
	Suppose we have an instance $\Omega$ of $\hol{\HWin{B},\HWeq{1}}$.
	We construct a new signature grid $\Omega'$ by replacing all $\HWeq{1}$ vertices with $\HWin{\{0,1\}}$ vertices.
	Let this set of $\HWin{\{0,1\}}$ vertices be $U$.
	We note that $\Hol{\Omega'}$ can be thought of as the summation of $2^{\abs{U}}$ Holants, where each vertex in $U$ is assigned to either $\HWeq{0}$ or $\HWeq{1}$.
	Define $A_i$ as the partial summation of the Holants on $\Omega'$ where $i$ vertices from $U$ are assigned $\HWeq{1}$
	and the rest are assigned $\HWeq{0}$.
	That is, we have that
	\(
		\Hol{\Omega'} = \sum_{i=0}^{\abs{U}} A_i.
	\)
	
	Now for an integer $d\geq 1$, we construct a new graph $\Omega'_d$ by attaching a length-$d$ path of $\HWin{\{0,1\}}$ to every node in $U$.
	Let $P_1(d)$ be the number of solutions when the leading edge of the path is selected,
	and similarly let $P_0(d)$ be the number of solutions when the leading edge of the path is not selected.
	Then we have
	\[
		\Hol{\Omega'_d} = \sum_{i=0}^{\abs{U}} A_i (P_0(d))^i (P_0(d)+P_1(d))^{\abs{U}-i}
		= (P_0(d))^{\abs{U}}\sum_{i=0}^{\abs{U}} A_{\abs{U}-i} \biggl( \frac{P_0(d)+P_1(d)}{P_0(d)}\biggr)^{\!\!i}.
	\]
	We claim that by interpolation on $d$, we can recover the values of $A_i$ for any $i$.
	In particular, $A_{\abs{U}}$ will correspond to $\Hol{\Omega}$.
	To show this, we only need to argue that
	$\frac{P_0(d)+P_1(d)}{P_0(d)}$ will take at least $\abs{U}$ unique values,
	and that these are computable in polynomial time.
	Now, since $P_0,P_1$ can be defined for any integral path length $d$, we have the relation
	\begin{equation*}
		P_0(d) = P_1(d-1) + P_0(d-1)
		\qquad
		\text{and}
		\qquad
		P_1(d) = P_0(d-1).
	\end{equation*}
	Then, applying \cref{prop:count:distinct-ratio} with $M=\begin{smallbmatrix}
		1 & 1\\
		1& 0
	\end{smallbmatrix}$ and $U=\begin{smallbmatrix}
		1\\
		1
	\end{smallbmatrix}$,
	and since $P_0(d), P_1(d) >0$ for every $d$, we get the required statement.
	Since we are only attaching a path to certain vertices,
	this does not alter the pathwidth or cutwidth by more than 1. 
\end{proof}

\begin{lemma}\label{lem:count:hol-to-genfactor-forced}
	For every finite list $B\subseteq \SetN$,
	there is a polynomial time many-one reduction from
	\begin{itemize}
		\item $\hol{\HWin{B},\HWin{\{0,1\}}}$ if $\max B -1 \in B$, 
		\item $\hol{\HWin{B},\HWeq{1}}$ if $\max B -1 \not\in B$, 
	\end{itemize}
	to $\hol{\HWin{B},\HWeq[1]{1}}$ increasing pathwidth and cutwidth by a fixed constant
	and increasing the degrees of the $\HWin{\{0,1\}}$ or $\HWeq{1}$ nodes by at most $\max B$.
\end{lemma}
\begin{proof}
	Consider any instance $\Omega$ of $\hol{\HWin{B},\HWeq{1}}$ or $\hol{\HWin{B},\HWin{\{0,1\}}}$.
	It suffices to alter the nodes labeled with $\HWeq{1}$ or $\HWin{\{0,1\}}$:
	We replace each of them by a new node with relation $\HWin{B}$
	and force $\max B-1$ edges by adding the same number of pendant nodes with relation $\HWeq[1]{1}$.
	Depending on whether $\max B-1 \in B$ or not, we get a $\HWin{\{0,1\}}$ or a $\HWeq{1}$ node.
	Since we are adding $\max B-1$ pendant nodes to certain vertices, the pathwidth increases at most by 1 and the cutwidth increases by at most $\max B-1$.  
\end{proof}
At this stage, we only need to realize $\HWeq[1]{1}$ nodes.
We will see that this is possible, but that there are some caveats.
Like before, in the case where $B$ only has even integers,
some parity issues force us to be able to realize only an even number of $\HWeq[1]{1}$ nodes.
This can be seen as realizing $\HWeq[2]{2}$ nodes instead. 

\begin{lemma}\label{lem:count:genfactor-forced-to-genfactor-forced2}
	For every fixed, finite list $B\subseteq \SetN$,
	there is a polynomial time many-one reduction
	from $\hol{\HWin{B},\HWeq[1]{1}}$
	to $\hol{\HWin{B},\HWeq[2]{2}}$
	increasing pathwidth and cutwidth only by a fixed constant
	and leaving the maximum degree unaffected,
	or increasing it to at most $\max B$.
\end{lemma}

\begin{proof}
	We claim that without loss of generality, we can assume that we have an even number of $\HWeq[1]{1}$ nodes. Suppose not. Then we consider two cases.
	\begin{itemize}
		\item If $B$ has some odd number $k$,
		then we add a separate component with one $\HWin{B}$ node connected to $k$ $\HWeq[1]{1}$ nodes.
		This brings the total number of $\HWeq[1]{1}$ nodes to be even.
		As this component is a star graph,
		pathwidth does not change.
		Cutwidth increases by at most $\max B$.
		\item If $B$ only has even numbers, then we claim that $\Hol{G}$ is zero.
		Consider the sum of the degrees of the nodes in any solution.
		The nodes with relation $\HWeq[1]{1}$ contribute an odd number,
		whereas the other nodes with relation $\HWin{B}$ can only contribute an even number.
		Thus, their sum must be odd.
		This contradicts the fact that the sum of the degrees of the vertices in any graph must be even.
	\end{itemize}

	\newextmathcommand{\intr}{\operatorname{int}}
	Let $X_1,X_2,X_3,\dots$ be the bags of some nice path decomposition of $G$,
	i.e.\ there is especially a unique introduce node for every vertex.
	Let $U=\{v_1,\dots,v_{\abs{U}}\}$ be the set of $\HWeq[1]{1}$ nodes of $G$,
	such that $\intr(i) < \intr(i+1)$ for all $i$,
	where $\intr(i)$ is the index of the bag introducing $v_i$.
	For all $i\in[\abs{U}]$,
	replace $v_i$ by a $\HWeq[2]{2}$ node keeping the incident edge.
	Then connect $v_{2i-1}$ and $v_{2i}$ by an edge.
	We add $v_{2i-1}$ to the bags
	$X_{\intr(2i-1)+1},\dots,X_{\intr(2i)}$.
	By assumption, there is no $v_\ell$ introduced in these bags
	and hence, the pathwidth increases by at most 1.
	Further, the degree of the graph does not increase
	as we can assume w.l.o.g.\ that it is already at least 2.

	The result directly transfers to cutwidth
	when using a linear layout instead.
\end{proof}
The following lemma completes the chain of reductions by handling $\HWeq[2]{2}$ relations. For the proof, we need to consider different cases depending on whether $B$ contains $0$, $1$, or some odd number. 

\begin{lemma}\label{lem:count:genfactor-forced2-to-genfactor}
	Let $B\subseteq \SetN$ be a fixed finite set.
	There is a polynomial-time Turing reduction from $\hol{\HWin{B},\HWeq[2]{2}}$ to $\hol{\HWin{B}}$ increasing pathwidth and cutwidth only by a fixed constant and leaving the max degree unaffected, or increasing it to $2\max B+6$.
\end{lemma}
\begin{proof}
	 If $0 \notin B$,
	 we can use the construction from \cref{lem:dec:all} to get a $\HWeq[2]{2}$ node. 
	 For the case when $0\in B$, we do a case-by-case analysis depending on $B$.
	 In either case, we attach a subgraph with a constant pathwidth and cutwidth to vertices. This does not affect either of them by more than $2\max B+ 6$, a fixed constant. 
	\paragraph*{Case 1: $B$ contains 1}
	Define $m\ge 2$ to be the smallest integer not in $B$.
	Consider the gadget in \cref{fig:count:path-interpolaton}.
	\begin{figure}
		\centering
		\includegraphics{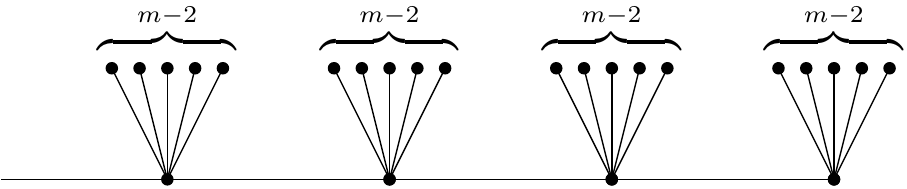}
		\caption{The gadget for case 1: Black nodes are $\HWin{B}$ nodes.}
		\label{fig:count:path-interpolaton}
	\end{figure}
	Suppose there are $d$ such vertices with $m-2$ pendant nodes each.
	Let all of them have the relation $\HWin{B}$.
	Let $P_1(d)$ be the number of solutions of the gadget where the dangling edge is selected in the solution.
	Similarly define $P_0(d)$ when the dangling edge is not selected.
	We claim that the gadget described can be effectively used to force edges,
	i.e.\ a $\HWeq[1]{1}$ node.
	Two such gadgets will give us a $\HWeq[2]{2}$ node.
	Suppose any graph $G$ contains $t$ such gadgets.
	We have
	\begin{align*}
		\Hol{G} = \sum_{i=0}^{t} A_i (P_0(d))^{t-i} (P_1(d))^{i} = (P_0(d))^{t} \sum_{i=0}^{t} A_i  \biggl( \frac{P_1(d)}{P_0(d)}\biggr)^{\!\!i}
	\end{align*}
	where $A_i$ is the number of ways of extending the solution in $G$ when $i$ of the gadgets choose to match their dangling edge. Through standard interpolation techniques, we can recover the $A_i$s, and thus $A_t$ will give us the solution where each gadget behaves like a $\HWeq[1]{1}$. Now, we can replace $\HWeq[2]{2}$ nodes in the $\hol{\HWin{B},\HWeq[2]{2}}$ instance with pairs of $\HWeq[1]{1}$ nodes. 
	
	To argue that we can do the interpolation, we need to show
	that $\frac{P_1(d)}{P_0(d)}$ will take at least $t$ unique values,
	and that these are computable in polynomial time. Since we can define such a gadget for any integer $d$ we have
	\[
		P_0(d)= kP_0(d-1) + kP_1(d-1)
		\qquad
		\text{and}
		\qquad
		P_1(d)= kP_0(d-1) + (k-1)P_1(d-1)
	\]
	for $k=2^{m-2}$.
	We now apply \cref{prop:count:distinct-ratio} with $M=\begin{smallbmatrix}
		k & k\\
		k& k-1
	\end{smallbmatrix}$ and $U=\begin{smallbmatrix}
		k\\
		k
\end{smallbmatrix}$.

\paragraph*{Case 2: $B$ does not contain 1, but contains some odd number} 
	In this case, we first observe that we can realize $\EQ{2}$ nodes.
	Let $m$ be the smallest non-zero number in the list $B$.
	Then we make an $m+1$-clique of vertices with signature $\HWin{B}$ and remove one edge, say $(u,v)$.
	We add one dangling edge each to $u,v$.
	This gives us a $\EQ{2}$ node.
	We construct the gadget in \cref{fig:count:path-interpolaton-general} to get a $\HWeq[1]{1}$ node. 
	Note that any double edges can be removed by placing $\EQ{2}$ nodes on them.

	\begin{figure}
		\centering
		\includegraphics{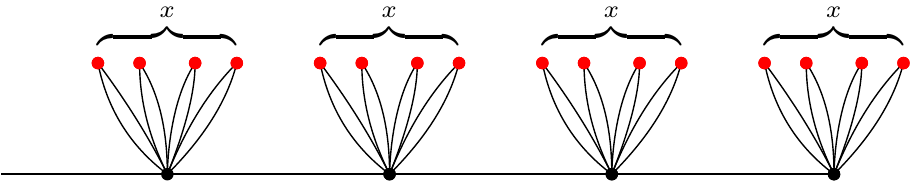}
		\caption{Gadget for case 2:
		Red nodes are $\EQ{2}$ nodes and black nodes are $\HWin{B}$ nodes.}
		\label{fig:count:path-interpolaton-general}
	\end{figure}

	 We will fix the value of $x$ later. Suppose there are $d$ $\HWin{B}$ nodes. Let $P_1(d)$ be the number of solutions of the gadget where the dangling edge is selected in the solution. Similarly define $P_0(d)$ when the dangling edge is not selected. We claim that this can be used to force edges. Suppose any graph $G$ contains $t$ such gadgets. Then like before, we have
	\begin{align*}
		\Hol{G} = \sum_{i=0}^{t} A_i (P_0(d))^{t-i} (P_1(d))^{i} = (P_0(d))^{t} \sum_{i=0}^{t} A_i  \biggl( \frac{P_1(d)}{P_0(d)}\biggr)^{\!\!i}
	\end{align*}
	where $A_i$ is the number of ways of extending the solution in $G$ when $i$ of the gadgets choose to match their dangling edge.
	We need to show that $\frac{P_1(d)}{P_0(d)}$ will take at least $t$ unique values,
	and that these are computable in polynomial time. 
	
	We first define the following quantities. 
	\begin{align*}
		F_0 = \sum_{i:2i\in B} \binom{x}{i} \qquad
		F_1 = \sum_{i:2i+1\in B} \binom{x}{i} \qquad
		F_2 = \sum_{i:2i+2\in B} \binom{x}{i}.
	\end{align*}
	Note, when these quantities are interpreted as a polynomial in $x$,
	each has degree at most $\lfloor \frac{\max B}{2} \rfloor$.
	For $j\in \{0,1,2\}$, $F_j$ can be interpreted as the number of ways a $\HWin{B}$ node can choose $\EQ{2}$ nodes
	given that $j$ of its adjacent edges in the horizontal path are selected.
	Since $B$ contains some odd integer and some even integer (since $0\in B$), we immediately get that $F_0,F_1\geq 1$.
	This ensures that $P_0(d),P_1(d)>0$ for every $d$.
	Then we have that the gadget follows the recurrence relation
	\[
		P_0(d) = F_0P_0(d-1) +F_1 P_1(d-1)
		\qquad
		\text{and}
		\qquad
		P_1(d) = F_1P_0(d-1) + F_2P_1(d-1).
	\]
	To apply \cref{prop:count:distinct-ratio}, we need to show that $M=\begin{smallbmatrix}
		F_0 & F_1\\
		F_1 & F_2
	\end{smallbmatrix}$ is invertible and that $U = \begin{smallbmatrix}
	F_0\\
	F_1
\end{smallbmatrix}$ is not an eigenvector of $M$.
To argue that $M$ is invertible, we can consider 
\(
	F_0F_2 - (F_1)^2= 0
\)
as a polynomial in $x$.
Since it has degree at most $\max B$, it can have at most that many solutions.
We pick $x$ to be the smallest positive integer that is not a solution to the above polynomial.
Then we have that $x\leq \max B +1$.  
Now, to show that $U$ is not an eigenvector of $M$, we show a contradiction. Consider the equation 
\[
	\begin{bmatrix}
		F_0 & F_1\\
		F_1 & F_2
	\end{bmatrix} \begin{bmatrix}
		F_0\\
		F_1
	\end{bmatrix} = \lambda \begin{bmatrix}
	F_0\\
	F_1
  \end{bmatrix}
	\qquad
	\implies
	\qquad
	\begin{matrix}
		(F_0)^2 + (F_1)^2 &= \lambda F_0\\
		(F_0 + F_2)F_1 &= \lambda F_1
	\end{matrix}
\]
Since $F_1>0$, we get $\lambda = F_0+F_2$ from the last equation. This gives
\[
 (F_0)^2+ (F_1)^2 = F_0(F_0+F_2)
 \qquad
 \implies
 \qquad
 (F_1)^2 =F_0F_2.
\]
This is a contradiction since we already chose $x$ such that $F_0F_2 - (F_1)^2 \neq 0.$ 

\paragraph*{Case 3: $B$ does not contain any odd number}
We can use similar arguments as for the previous case,
except that we modify the gadget.
We replace the single edges with double edges to get the gadget from \cref{fig:count:path-interpolaton-even}. 
Like before, we can construct a  $\EQ{2}$ node in this case and we can 
similarly remove any double edges by placing $\EQ{2}$ nodes on them.
\begin{figure}
	\centering
	\includegraphics{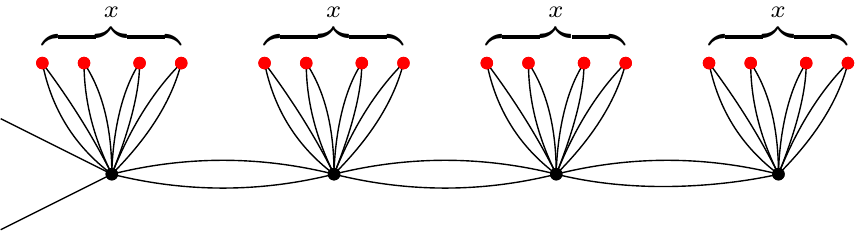}
	\caption{Gadget for case 3:
	Red nodes are $\EQ{2}$ nodes constructed like before and black nodes are $\HWin{B}$ nodes.}
	\label{fig:count:path-interpolaton-even}
\end{figure}

The only thing that needs to be taken care of
is the possibility that a solution might select one of the double edges of the horizontal path.
But this is not possible
since $B$ does not contain any odd number
and thus the terminal node must be incident to an even number of selected edges.
Similarly, there is no solution where exactly one of the dangling edges is selected.

Let $P_2(d)$ be the number of solutions of the gadget where both the dangling edge are selected in the solution. Defining $P_1(d),P_0(d)$ the same way as before, we can argue that $P_1(d)=0$ for any $d$. We can interpolate on $\frac{P_2(d)}{P_0(d)}$. To show that we can do this, define $F_0, F_2$ similarly. Let
\begin{align*}
F_4 &= \sum_{i:2i+4\in B} \binom{x}{i}.
\end{align*}
Then we have
	\[
	P_0(d) = F_0P_0(d-1) +F_2 P_2(d-1)
	\qquad
	\text{and}
	\qquad
	P_2(d) = F_2P_0(d-1) + F_4P_2(d-1).
\]
We can apply the same arguments as before with $M=\begin{smallbmatrix}
	F_0 & F_2\\
	F_2 & F_4
\end{smallbmatrix}$ and $U = \begin{smallbmatrix}
	F_0\\
	F_2
\end{smallbmatrix}$. 
\end{proof}

\begin{proof}[Proof of \cref{thm:count:bfr-to-gen-matching}]
	We prove this through a chain of reductions.
	Given any instance of \CountBFR,
	we can sequentially apply \cref{lem:count:BFR-to-hol-edge-weights,lem:count:hol-edge-weights-to-hol1,lem:count:hol1-to-hol01,lem:count:hol-to-genfactor-forced,lem:count:genfactor-forced-to-genfactor-forced2,lem:count:genfactor-forced2-to-genfactor}
	to get a polynomial number of instances of $\hol{\HWin{B}}$
	such that the pathwidth is affected only by some function of $\Delta^*$.
	As $\hol{\HWin{B}}$ is exactly \CountBFactor the theorem follows.
\end{proof}

\section{Lower Bound when Parameterizing by Cutwidth}\label{sec:lowerCW}
The algorithmic result from \cref{thm:algo:cutwidth:main}
shows that the pathwidth lower bound breaks when parameterizing by cutwidth.
Nevertheless, we can show that this ``improved'' running time is the best we can hope for assuming SETH and \#SETH.
For this we use the same high level ideas Curticapean and Marx presented in Figure~6 of \cite{CurticapeanM16}
where they reduce from \#SAT to computing the Holant and then reduce to counting perfect matchings.
But the construction can also be seen as a modification of our reduction for the pathwidth lower bound.
We again first reduce to the intermediate problem \BFR and then to \BFactor.
By this we can reuse the results of realizing relations that we have seen in the previous sections.

\begin{theorem}\label{thm:lowerCW:satToBFR}
  Let $B \subseteq \SetN$ be a fixed set of finite size.
  Given a CNF-formula $\phi$ with $n$ variables and $m$ clauses.
  We can construct a (simple) \BFR instance $G$ with $\O(nm)$ vertices,
  bounded degree
  and a linear layout of width $\cutw \le n+\O(1)$
  in time linear in the output size.
  Further, the number of solutions for $\phi$ is equal to the number of solutions for $G$.
\end{theorem}

\subsection{High Level Construction}\label{sec:lowerCW:construction}
Recall, that for the pathwidth lower bound we grouped variables together.
This was needed to keep the pathwidth of the construction low.
But this increased the cutwidth of the graph.
Now, we do not group variables together but encode each variable on its own.
See \cref{fig:lowerCW:construction} for an example of the construction we describe formally in the following.

\begin{figure}
    \centering
    \includegraphics{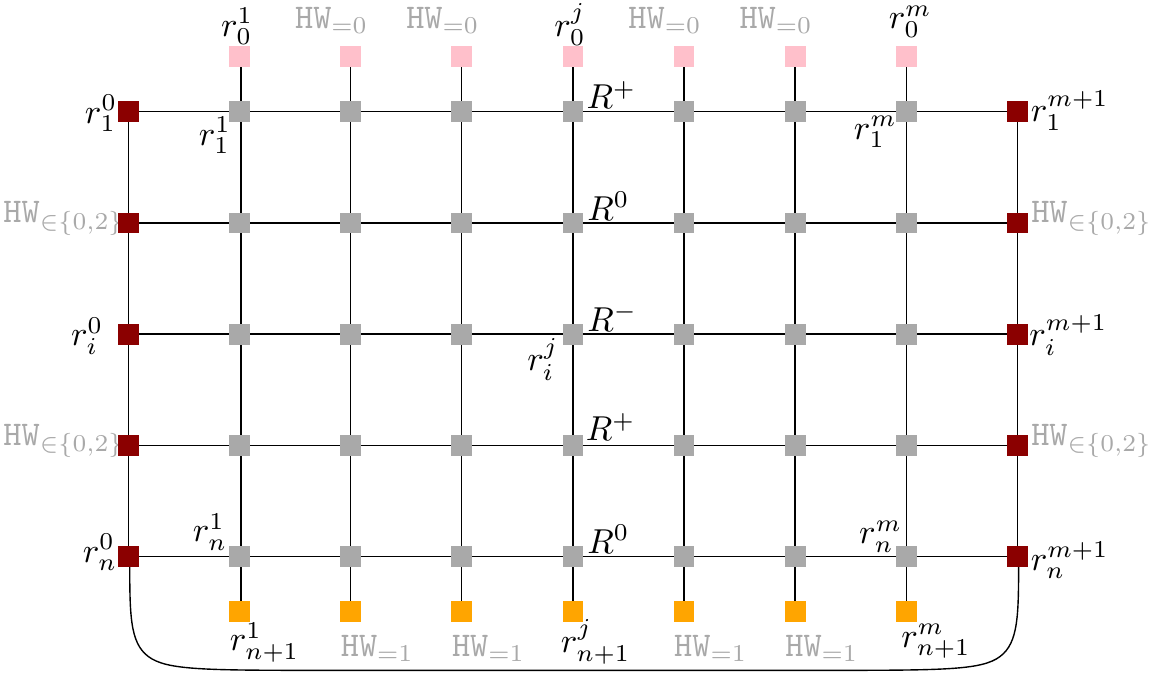}
    \caption{The example graph for a formula containing the clause $(x_1 \lor \bar x_3 \lor x_4)$.}
    \label{fig:lowerCW:construction}
\end{figure}
Let $x_1,\dots,x_n$ be the variables
and $C_1,\dots,C_m$ the clauses of $\phi$.
For each $i\in[n]$ and every $j\in[m]$ we create a vertex $r_i^j$.
We assign the relation $R^+$ to $r_i^j$ if $x_i$ appears positively in $C_j$,
$R^-$ if it appears negatively,
and otherwise $R^0$,
where $R^0$, $R^+$, and $R^-$ are defined later.
Additionally add vertices $r_i^0$ and $r_i^{m+1}$
with relation $\HWin{\{0,2\}}$ for all $i\in[n]$.
We say that the vertices $r_i^0,\dots,r_i^{m+1}$ form the $i$th \emph{row},
i.e.\ the row of variable $x_i$.
Create new nodes $r_0^j$ and $r_{n+1}^j$
and assign the relations $\HWeq{0}$ and $\HWeq{1}$ to them
for all $j\in[m]$, respectively.
We say the vertices $\{r_i^j\}_i$ form the $j$th \emph{column}.
We connect two nodes $r_i^j$ and $r_{i'}^{j'}$ by an edge
if $\abs{i-i'}\le1$ and $\abs{j-j'}\le 1$
for all $i,i'\in[0,n+1]$ and $j,j'\in[0,m+1]$,
i.e.\ if they are neighbors in the grid.

The idea is the same as for the pathwidth construction,
except that selecting the edges of the $i$th row
corresponds to setting the variable $x_i$ to true.
The edges between the nodes of a column represent
if a clause is already satisfied.
The relation $\HWeq{0}$ ensures that we start with an initially unsatisfied clause.
At each node $r_i^j$ we check
whether the assignment to this variable $x_i$ satisfies the clause $C_j$
and then force the output edge (i.e.\ the bottom edge) to be selected.
Otherwise we propagate the current state (i.e.\ the selection of edges).
Eventually we reach $r_{n+1}^j$ with relation $\HWeq{1}$
where the edge has to be selected and thus the clause must be satisfied.

The relations $R^0$, $R^+$, and $R^-$ accept exactly those inputs that satisfy all of the following conditions:
\begin{enumerate}
  \item
  The left edge is selected if and only if the right edge is selected.

  \item
  If the top edge is selected, the bottom edge is selected.

  \item
  Only for $R^+$:
  If the top edge is unselected and the left edge is selected,
  then the bottom edge is selected.

  \item
  Only for $R^-$:
  If the top edge is selected and the left edge is not selected,
  then the bottom edge is selected.
\end{enumerate}

\subparagraph*{Final Modifications of the Construction.}
As for the pathwidth lower bound this construction does not fulfill all requirements of a \BFR instance.
To avoid these problems, we first apply the same set of modifications we had for the pathwidth bound.
That is, adding a negated edge between $r_i^j$ and $r_{i+1}^j$ for all $i$ and $j$.
Further, we introduce the nodes $\hat r_0^j$ and $\hat r_{n+1}^j$ and connect them to the first and last node of a column.
Then we modify the relations that always the negated edge is selected if and only if the positive edge is not selected.
The last step merges the nodes $r_0^j$, $r_{n+1}^j$, $\hat r_0^j$, and $\hat r_{n+1}^j$ together into a node $r^j$ with appropriate relation.

But still the Hamming weight of the accepted inputs of the relations $R^0$, $R^+$, $R^-$, and $\HWin{\{0,2\}}$ is not equal.
Strictly speaking,
the Hamming weights of the inputs of these relations differs by exactly 2.
Hence, we add a loop to each node with such a relation.
This loop is counted as two inputs that are set to true
if the Hamming weight is too small.
One can use $\HWeq[4]{2}$ nodes to replace (parallel edges and) loops
to get a simple graph.

These modifications increase cutwidth only by additive terms.
As they are mostly only needed for technical reasons
we usually ignore them and only consider them when relevant for the proof.
It remains to show the correctness of the construction
and the bounds for the size and the cutwidth.
\begin{lemma}\label{lem:lowerCW:construction:completeness}
  If $\phi$ is satisfiable, then there is a solution for the \BFR-instance.
\end{lemma}
\begin{proof}
  By the construction and the definition of the relations it directly follows that there is a solution for the \BFR-instance.
  For this we choose for each clause the first variable that can satisfy it
  to select the bottom edge.
  Further observe that we put the edge from $r_n^0$ to $r_n^{m+1}$ in the solution
  if the number of variables set to true is odd.
\end{proof}
\begin{lemma}\label{lem:lowerCW:construction:correctness}
  If there is a solution to the \BFR instance, then $\phi$ is satisfiable.
\end{lemma}
\begin{proof}
  Observe that the edges in each row are either all selected
  or all not selected by the definition of $R^0$, $R^+$, and $R^-$.
  Hence, we can define a consistent assignment to the variables as follows:
  if the edges of the $i$th row are selected, set $x_i$ to true and otherwise to false.
  For each $j$ especially the relations at the vertices $r_0^j$ and $r_{n+1}^j$ are satisfied by the solution.
  Hence, the edges between the nodes of this column
  are neither all selected nor unselected.
  This implies that there is some $i\in[n]$
  such that the top edge of $r_i^j$ is not selected,
  but the bottom edge is selected.
  By definition, $r_i^j$ cannot be labeled with $R^0$.
  If the vertex is labeled with $R^+$, this change can only happen if the edges of the row are selected.
  But then our assignment satisfies clause $j$.
  The same argument holds for the case when the vertex has relation $R^-$ and the edges of the row are not selected.
  Since this holds for all $j$, all clauses and hence $\phi$ is satisfied.
\end{proof}
To obtain a tight lower bound we need to analyze the cutwidth of the graph.
\begin{lemma}\label{lem:lowerCW:construction:bounds}
  The graph has $\O(nm)$ nodes of constant degree.
  The graph has cutwidth at most $n + \O(1)$.
\end{lemma}
\begin{proof}
  To bound the cutwidth it suffices to give a linear layout with the stated width.
  For this we go through the graph column by column
  and enumerate the vertices from the top to the bottom.
  By this only the vertices of two columns have to be considered
  to determine the width of the cut.
  Each vertex has (at most) one edge to the other side of the cut
  (i.e.\ the edge to the next column).
  Only the last visited vertex can contribute more edges to the cut,
  but only constantly many.
  The merging of the four nodes into $r^j$ increases cutwidth at most by 4.
  Including the additional edge from $r_n^0$ to $r_n^{m+1}$
  that is present in almost every cut, we get the claimed bound.
\end{proof}
As we have seen there is a one-to-one correspondence between satisfying assignment to the formula and selection of edges in the graph.
For this to work,
we crucially need that we cannot choose which variable is ``responsible'' for satisfying the clause
as we always must choose the first such possibility in each column.
Hence we get the following corollary.
\begin{corollary}\label{corr:lowerCW:construction:parsimonious}
  The reduction is parsimonious,
  i.e.\ it preserves the number of solutions.
\end{corollary}
In the remaining part of the section we show the proof of the lower bounds
by combining this construction with the previous results from \cref{sec:dec,sec:opt,sec:count}.

\subsection{Decision Version}\label{sec:lowerCW:dec}
It suffices to combine the above reduction
with the realization results from \cref{sec:dec}.
\begin{proof}[Proof of \cref{thm:lowerCW:main}~(1)]
  We use the construction from \cref{thm:lowerCW:satToBFR} to transform any formula into a \BFR instance.
  Then we use \cref{thm:dec:realization} to replace every relation by its realization.
  This increases cutwidth at most by a constant factor.
  Executing the claimed algorithm with running time $\Ostar{(2-\epsilon)^{\cutw}}$ 
  on this graph
  directly contradicts SETH.
\end{proof}
This directly implies the analog lower bound for \MinBFactor.

\subsection{Optimization Version}\label{sec:lowerCW:opt}
Once more it suffices to only consider \MaxBFactor if $0\in B$.
For the case $0\notin B$
the hardness follows from the decision version.
We use the same approach as before.
Let $G$ be a \BFR instance obtained from \cref{thm:lowerCW:satToBFR}.
Replace all relations by their realization according to \cref{thm:opt:realization}
with penalty being twice the degree
and let the resulting \BFactor instance be $G'$.
We make use of the following lemma:
\begin{lemma}\label{lem:lowerCW:opt:helper}
  For all pairs $G,G'$ of graphs resulting from the above modification,
  there is an efficiently constructible constant $\gamma$ such that
  $G$ has a solution if and only if the largest solution for $G'$ has size $\gamma$.
\end{lemma}
\begin{proof}
  As $G$ only consists of complex nodes,
  we split each edge into two half-edges
  and assign both halves to their corresponding endpoint,
  to avoid counting edges twice.

  Now let $v$ be an arbitrary complex node in $G$.
  By the definition of realizations,
  there is an $\alpha_v$ such that any partial solution in $G'$
  that respects the relation of $v$
  (i.e.\ the correct edges are selected)
  has size exactly $\alpha_v$
  (as the Hamming weight of all accepted inputs is equal,
  the $\alpha_v$ can take care of the half-edges).
  Otherwise the size of this partial solution is at most $\alpha_v-\beta_v$
  where $\beta_v$ is the penalty for this realization.
  Recall, that we added loops to most nodes
  such that we can replace the relations by their realization.
  We define $\delta_v$ for each node to take care of this.
  For the nodes $r^j$ we have $\delta_v=0$
  while we set $\delta_v=1$ for all other nodes.
  We define $\gamma_v = \alpha_v + \delta_v$
  and set $\gamma \deff \sum_{v \in V(G')}\gamma_v$.
  By assumption the $\alpha$s are efficiently constructible,
  so is $\gamma$.

  The ``only if'' direction of the claim follows directly by the definition of $\gamma$.
  Now assume there is no solution for $G$.
  The solution for $G'$ cannot be larger than $\gamma$
  because then for some node $v$ the partial solution would be larger than $\gamma_v$
  which is not possible by the definition of a realization and $\alpha_v$.

  As $G$ has no solution,
  for each subset $S$ of selected edges,
  there must be at least one node $v$ with relation $R$ that is not satisfied.
  Assume without loss of generality there is just one such node.
  By assumption we can extend $S$ to a set $S' \subseteq E(G')$ that satisfies the relations at all nodes except for $v$.
  By the definition of a realization we know
  that any partial solution for the graph realizing $R$ must have size at most $\alpha_v-\beta_v$
  as the selected dangling edges are not valid for $R$.
  Hence, we can bound the size of the partial solution by $\alpha_v-\beta_v+\deg(v)$
  as all dangling edges could be selected.
  But from $\beta_v \ge 2\deg(v)$ we get
  that this is strictly smaller than $\gamma_v$.
  Hence, the total solution size is strictly smaller than $\gamma$.
\end{proof}
Now we have everything ready to prove the $\Ostar{(2-\epsilon)^\cutw}$ lower bound
for \MaxBFactor when $0\in B$, assuming SETH.
\begin{proof}[Proof of \cref{thm:lowerCW:main}~(2)]
  Use the construction from \cref{thm:lowerCW:satToBFR} to get a \BFR instance.
  Use \cref{thm:opt:realization} to replace every relation by a graph,
  where the penalty of each realization is at least twice the degree of the relation.
  As each complex node has constant degree,
  this replacement increases the size only by a constant factor depending only on $B$.
  Assume the claimed algorithm with running time $\Ostar{(2-\epsilon)^\cutw}$ exists
  and run it on this graph to get the size of the largest solution.
  We compare this to the constant from \cref{lem:lowerCW:opt:helper}
  and can decide the satisfiability of the formula
  which directly contradicts SETH.
\end{proof}

\subsection{Counting Version}\label{sec:lowerCW:count}
For the counting version we again make use of the chain of Turing-reductions from \BFR to \CountBFactor.
\begin{proof}[Proof of \cref{thm:lowerCW:main}~(3)]
  Use the construction from \cref{thm:lowerCW:satToBFR} to transform any formula into a \BFR instance.
  Then we use \cref{thm:count:bfr-to-gen-matching} to get polynomially many instances of \CountBFactor.
  Assume the claimed algorithm with running time $\Ostar{(2-\epsilon)^\cutw}$ exists and run it on these graphs.
  This yields a faster algorithm for \#SAT contradicting \#SETH.
\end{proof}

\appendix
\bibliography{bib}

\end{document}